\pdfoutput=1
\documentclass[letterpaper, 12pt]{article}
\usepackage{amssymb,amsmath,amsthm}
\usepackage{epsf}
\usepackage{times,bm,mathrsfs}
\usepackage{amsmath}
\usepackage{amssymb}
\usepackage{amsfonts}
\usepackage{mathrsfs}
\usepackage{bbm}
\usepackage{color}
\usepackage{graphicx}
\usepackage{amscd}
\usepackage{epsfig}

\allowdisplaybreaks

\definecolor{Red}{rgb}{1,0,0}
\definecolor{Blue}{rgb}{0,0,1}
\definecolor{Olive}{rgb}{0.41,0.55,0.13}
\definecolor{Green}{rgb}{0,1,0}
\definecolor{MGreen}{rgb}{0,0.8,0}
\definecolor{DGreen}{rgb}{0,0.55,0}
\definecolor{Yellow}{rgb}{1,1,0}
\definecolor{Cyan}{rgb}{0,1,1}
\definecolor{Magenta}{rgb}{1,0,1}
\definecolor{Orange}{rgb}{1,.5,0}
\definecolor{Violet}{rgb}{.5,0,.5}
\definecolor{Purple}{rgb}{.75,0,.25}
\definecolor{Brown}{rgb}{.75,.5,.25}
\definecolor{Grey}{rgb}{.5,.5,.5}
\definecolor{Black}{rgb}{0,0,0}

\title{Long-branch attraction in species tree estimation: inconsistency of partitioned likelihood and topology-based summary methods}
\author{
	S\'ebastien Roch\footnote{Department of Mathematics, University of Wisconsin--Madison, 480 Lincoln Dr, Madison WI 53706} 
	\and
	Michael Nute\footnote{Department of Statistics, The University of Illinois at Urbana-Champaign, 725 S Wright St \#101, Champaign IL 61820}
	\and 
	Tandy Warnow\footnote{Department of Computer Science, The University of Illinois at Urbana-Champaign, 201 North Goodwin Avenue,
		Urbana IL 61801-2302} 
}
\date{\today}

\newtheorem{thm}{Theorem}

\newtheorem{cor}{Corollary}
\newtheorem{lem}{Lemma}

\newtheorem{claim}{Claim}

\newcommand{\tree}{\mathcal{T}}

\newcommand{\E}{\mathbf{E}}
\renewcommand{\P}{\mathbf{P}}

\newcommand{\ta}{\phi_a}
\newcommand{\tb}{\phi_b}

\newcommand{\stree}{\mathcal{S}}

\definecolor{Red}{rgb}{1,0,0}
\definecolor{Blue}{rgb}{0,0,1}
\definecolor{Grey}{rgb}{0.5,0.5,0.5}

\renewcommand{\patterns}{\mathscr{X}}
\newcommand{\patsat}{\patterns^{\mathrm{s}}}
\newcommand{\qsat}{Q^{\mathrm{s}}}
\newcommand{\bigo}{\mathcal{O}}

\begin{document}	
\maketitle

\begin{abstract}
With  advances in sequencing technologies, there are
now massive am\-ounts of genomic data from across all life,
leading to the possibility that a robust Tree of Life can be
constructed.
However, ``gene tree heterogeneity", which is when
different genomic regions can evolve differently, is a
common phenomenon in multi-locus datasets, and reduces the
accuracy of standard
methods for species tree estimation
that do not take this heterogeneity into account.
New methods have been developed for species tree estimation
that specifically address gene tree heterogeneity,
and that have been proven to converge to the true species tree
when the number of loci and number of sites per locus both increase (i.e., the methods are said to be ``statistically consistent").  
Yet, little is known about the biologically realistic condition
where the number of sites per locus is bounded.
We show that when the sequence length of each locus is bounded (by any arbitrarily chosen value), the most common approaches to species tree estimation that take heterogeneity into account (i.e., traditional fully partitioned concatenated maximum likelihood and newer approaches, called summary methods, that estimate the species tree by combining gene trees)  
are not statistically consistent,   {\em even when} the heterogeneity is extremely constrained.
The main challenge is the presence of conditions such as long branch attraction that create biased tree estimation when the number of sites is restricted. 
Hence, our study uncovers a fundamental challenge to species tree estimation using both traditional and new methods.
\end{abstract}

\section*{Introduction}

Species trees are a key aspect of much biological research, including the detection of co-evolution, the inference of the ancestral traits, and the dating of speciation events \cite{Posada2016}.
The  availability of sequence data collected from diverse species representing a broad spectrum of life  has led to the expectation that the construction of a robust Tree of Life should be possible using statistical estimation methods,  such as maximum likelihood.
These estimations are  increasingly based on large numbers of loci (sometimes thousands) selected from across the genomes of different species \cite{Meredith2011,jarvis-2014b,Misof-2014,wickett2014,Cannon2016,DavidMaddison2016}.

By and large, however, the methods used for species tree estimation have been designed for  gene tree estimation, which is a simpler statistical estimation problem.
For gene tree estimation, the assumption is that the input sequences have all evolved down a single model tree (called the ``gene tree'') under a sequence evolution model, such as
Cavender-Farris-Neyman \cite{Cavender1978,Farris1973,Neyman1971},
Jukes-Cantor \cite{Jukes1969},
or the Generalised Time Reversible (GTR) model \cite{tavare-gtr}.
The estimation of the gene tree under these models from the aligned sequence data 
is a well-studied problem, and many statistically
consistent methods have been developed under
these models \cite{semple2003phylogenetics}.
Species tree estimation is much more complex,
since gene trees can differ from the species tree
due to multiple causes, including
incomplete lineage sorting (ILS), as modelled
by the multi-species coalescent (MSC) model \cite{Maddison1997a}.
Indeed, many recent phylogenetic analyses of genome-scale biological datasets for birds \cite{jarvis-2014b}, land plants \cite{wickett2014}, worms \cite{Cannon2016}, and other organisms, have revealed substantial heterogeneity across the genes that is consistent with ILS.

The  construction of the species tree when there is gene tree heterogeneity due to ILS can be seen as a statistical estimation problem under a two-phase model of sequence evolution where gene trees evolve within a species tree under the MSC model, and then gene sequences evolve down each gene tree under a sequence evolution model.
For example, 
under the MSC+JC model where
true gene trees evolve within the species tree under
the MSC model and 
gene sequences evolve down the gene trees under
the Jukes-Cantor (JC) model,  
the estimation of species trees from
gene sequence data needs to use
the properties 
of the evolutionary models in order to be
statistically consistent.
One such approach for species tree estimation is to
estimate gene trees for each locus,
and then combine these gene trees into a species tree using
a coalescent-based summary method (that 
takes gene tree incongruence due to ILS into account); such
approaches can be
proven to converge in probability to the
true species tree as the number of genes and number of
sites per gene both increase.
Thus, for example, statistically consistent 
species tree estimation is possible under the MSC+JC model
when gene trees are estimated using Jukes-Cantor maximum likelihood 
and then combined into a species tree using
an appropriate coalescent-based summary method.
Examples of these summary methods that enable statistically
consistent species tree estimation include
MP-EST \cite{Liu2010a}, NJst \cite{njst}, ASTRID
\cite{astrid}, ASTRAL \cite{astral,astral-2},
STEM \cite{stem}, STEAC \cite{star}, STAR \cite{star},
and
GLASS \cite{Mossel2010}.

In contrast, 
many species trees are estimated using ``unpartitioned maximum
likelihood", where 
the gene sequence alignments are 
concatenated into a single supermatrix,  and a tree
is then estimated on that supermatrix under the assumption
that 
all the sites evolve under the same model tree. 
As shown by \cite{RochSteel2015}, this
approach is not statistically consistent and
can even be positively misleading in the presence of
gene tree  heterogeneity due to ILS.

Although unpartitioned concatenated analysis with maximum likelihood (CA-ML) 
is known to be statistically inconsistent and coalescent-based
species tree methods can be statistically consistent,
performance in practice (and in particular on 
simulated datasets) has been mixed, with 
CA-ML sometimes more accurate than leading summary methods \cite{leache2010accuracy, patel2013,Mirarab2014, WSB, chou-recombcg,molloy2017}.
One of the challenges to using summary methods is
gene tree estimation error, resulting in part from
limited sequence lengths per gene \cite{naive-binning}. 
The ``statistical binning" approach \cite{Mirarab2014} was designed
to improve the accuracy of species trees
estimated using summary methods by 
binning sequences from different genes together using 
statistical techniques for detecting
strongly supported incongruence (e.g., using 
bootstrap support on estimated gene trees)
and then estimate new gene trees on the combined
datasets. 
As shown in \cite{WSB}, weighted statistical binning (an improved version of the original statistical binning approach) 
followed by appropriate summary methods
is statistically consistent under the MSC+JC model. 

Note however that the guarantees of statistical consistency provided so far
have nearly always made the following 
assumptions: every locus is recombination-free,
the number of sites
per locus increases without bound, and the number of loci increases
without bound. 
These assumptions are unrealistic, since
recombination-free loci are generally short. 
Therefore, of greater relevance to practice is the
question of statistical consistency where
the number of recombination-free loci increases, but 
the number of sites
per locus is bounded by some $L \in \mathbf{Z}_+$
\cite{Warnow-PC2015,RochWarnow2015}.
We investigate this question for the following methods:

\begin{itemize}
	\item 
	fully partitioned  maximum likelihood,
	\item 
	topology-based
	summary methods  (i.e., methods that combine gene tree topologies),
	and
	\item 
	weighted statistical binning pipelines
	followed by topology-based summary methods.
\end{itemize}

We address this question under the MSC+CFN model,
where the CFN is the symmetric two-state sequence
evolution model (i.e., the two-state version 
of the Jukes-Cantor model); the results
we find extend to nucleotide sequence evolution models, but the
proofs are simplest under the CFN model.
Perhaps surprisingly, our results are negative: for
all $L$, none of the
approaches is statistically consistent under the MSC+CFN model and can even be positively misleading.
Furthermore, this problematic behavior occurs {\em even when} all the genes evolve down a single model CFN tree. 
Therefore, expectations of 
accurate species trees using any of these methods 
given large amounts of data
may be unfounded. 

The key challenge to species tree estimation is  {\em long branch attraction}, a phenomenon that can confound maximum likelihood tree estimation when sequence lengths for each genomic region are finite.
In fact, we show that many species tree estimation methods that are statistically consistent when the number of genomic regions and their lengths both increase become inconsistent when only the number of regions increases, and the sequence length for each genomic region is bounded (however arbitrarily).
These results suggest that all common approaches to species tree estimation are far from being mathematically rigorous, even under highly simplified model conditions where there is no heterogeneity between the loci.
This is a very substantial limitation for multilocus phylogeny estimation
methods in general, and shows that new
approaches for species tree estimation method are needed.

\section*{Multi-locus evolution under the MSC}\label{sec:basics}
Our analysis is based on the MSC+CFN model.
A CFN model tree is a binary tree $(\tree,\Lambda)$ with topology $\tree$ and branch lengths $\Lambda$.
Under the assumption that the tree has $n$ leaves, 
each site (character) $\chi$ refers to the length-$n$ vector of character states corresponding the same homologous site for each taxon. The possible character states are $\{0,1\}$ and evolutionary changes are modeled by a continuous-time Markov process with instantaneous rate matrix $Q = \begin{pmatrix}
-1/2 & 1/2\\
1/2 & -1/2
\end{pmatrix}$.
In particular, the probability of a change along a branch of length $\lambda$ is parametrized as $p=\frac{1}{2}\left(1-e^{-2\lambda}\right)$.
Under the MSC+CFN model, each locus $j$
evolves independently 
on a random gene tree $(\tree_j,\Lambda_j)$, which
is derived from the
multispecies coalescent on a species
tree $(\stree,\Gamma,\theta)$, where the
$\Gamma_e$s are the branch lengths
in units of $\theta_e = 2 N_e \mu_e$ with 
$N_e$ and $\mu_e$ the effective population
size and mutation rate of branch $e$. 
That is, on each branch $e$ of $\stree$, looking backwards in time, lineages entering the branch coalesce at rate $2/\theta_e$ according to the Kingman coalescent. The remaining lineages at the top of the branch enter the ancestral population, and so on.

We assume  that all $m$ loci evolve on the same species tree and that each locus has a constant, finite sequence length $L$. Let $\chi_{ij}$ represent site $i$ on locus $j$, where $1\le i \le L$ and $1\le j \le m$, and let $\chi_{\cdot j}$ represent the set of all characters for locus $j$. 
We refer to the $\chi_{\cdot j}$ as {\em $j$-th locus sequences}. 
Denote the entire set of characters on all loci as $X$.

\paragraph*{Inconsistency of partitioned maximum likelihood}\label{sec:fpml-msc}


Let $\mathcal{L}(\tree^0, \Lambda, \chi)$ denote the likelihood function for a single site $\chi$ under the CFN model on $(\tree^0, \Lambda)$, and let $\ell = \log \mathcal{L}$ be the log-likelihood. 
Under fully partitioned maximum likelihood, we seek a single binary tree topology $\tree^0$ but allow each locus to have its own branch length parameter $\Lambda_j$; hence, the general likelihood function over all sites and all loci is
\begin{eqnarray*}
	\ell^*(\tree^0,\Lambda_1,\dotsc,\Lambda_m, X)
	&=&\sum_{j=1}^m\sum_{i=1}^L \ell(\tree^0,\Lambda_j,\chi_{ij}),
\end{eqnarray*}
and a maximum likelihood topology is any element of the set
\begin{equation}
\label{eq:def-partML}
\underset{\tree^0}{\arg\max}\,
\underset{\Lambda_1,\ldots,\Lambda_m}{\max}\,\ell^*(\tree^0,\Lambda_1,\ldots,\Lambda_m, X).
\end{equation}

\begin{thm}[Inconsistency of partitioned ML]
	\label{thm:incPart}
	Under the MSC+CFN model, fully partitioned maximum likelihood on loci with a bounded number of sites is not statistically consistent and is even positively misleading.
	That is, for any length $L\in\mathbb{N}$, there is a species tree with topology, branch lengths and mutation rates such that, given data generated under the MSC+CFN model, as the number of loci $m\rightarrow\infty$, the maximum likelihood topology is unique and is different from the true species tree topology with probability going to 1.
\end{thm}
\noindent The proof of this theorem is provided below.

\paragraph{Inconsistency of topology-based summary methods}
Summary methods have been developed that are designed to address heterogeneity between gene tree topologies due to ILS, and are statistically consistent under the MSC model.
We  consider topology-based summary methods that take as input unrooted gene trees, and only use their topologies and not any additional information (e.g., sequence data, branch lengths, bootstrap support).
\begin{itemize}
	\item We assume that the tree provided for a given gene sequence alignment is its maximum likelihood gene tree, and if there is a tie for the best maximum likelihood tree topology, then a random best-scoring tree is selected. 
\end{itemize}

When the number of species is four, then the summary method is selecting the best unrooted tree topology from the three possible unrooted tree topologies, also referred to as quartet trees. 
By \cite{Allman}, under the MSC the most probable quartet tree is the true species tree for any four species (i.e., there is no anomaly zone on unrooted four-leaf species trees).
Hence, in the four species case, we will make the assumption that the summary method will return the tree topology that appears the most
frequently among its input gene trees, as this is
a statistically consistent technique for estimating the unrooted
species tree on four leaves.
We refer to this most frequent quartet tree as the ``dominant" quartet tree.
That is, we
restrict ourselves to the following ``reasonable'' property of
a summary method
$\mathcal{A}$:
\begin{itemize}
	\item
	When $n=4$, as the number of loci $m$
	increases then with probability converging to $1$,
	$\mathcal{A}(\mathcal{T}_1,\ldots,\mathcal{T}_m)=t$ where $t$ is
	the quartet tree that appears with the highest frequency in the
	input $\mathcal{T}_1,\ldots,\mathcal{T}_m$;  if there
	are ties, then $\mathcal{A}$ picks uniformly at random between the
	most frequent quartet trees.
\end{itemize}
We will say that the summary method $\mathcal{A}$ is
{\em reasonable} if it satisfies
this property.
Many of the popular
summary methods (e.g., ASTRAL and BUCKy) are reasonable in that sense.
\begin{thm}[Inconsistency of reasonable summary methods]
	\label{thm:incReas}
	Under the MSC+CFN model, any reasonable summary method $\mathcal{A}$ with maximum likelihood input trees on loci with a bounded number of sites is not statistically consistent.
	That is, for any length $L\in\mathbb{N}$, there is a species tree with topology, branch lengths, and mutation rates, such that given data generated under the MSC+CFN model, as the number of loci $m\rightarrow\infty$, the topology produced by $\mathcal{A}$ is unique and is different from the true species tree topology with probability going to 1. 
\end{thm}

\paragraph*{Inconsistency of weighted statistical binning followed by a summary method}
The ``statistical binning" method, and its improved version ``weighted statistical binning'', were developed to address challenges in species tree estimation that result from gene tree estimation error. 
In \cite{WSB} it was shown that statistical binning was inconsistent under the MSC+CFN model but that weighted statistical binning (WSB) was statistically consistent.
Those proofs depend crucially on the number of sites per locus increasing to infinity, and so this previous work did not address the case we consider here, where each site has length bounded by $L$.

In a WSB pipeline, estimated gene trees with bootstrap support are provided for every locus, and then an incompatibility graph is computed for that set of gene trees with branch support.
The graph is used to partition the genes into sets (called ``bins") and then ``supergene trees" are computed using a fully partitioned maximum likelihood analysis on each bin. 
These supergene trees are then given to the selected summary method as input, and a species tree is returned.
In a weighted statistical binning pipeline, each supergene tree is replicated by the number of genes in its associated bin.
The incompatibility graph depends on a parameter $B$, as follows: two gene trees are considered to be incompatible if there is a pair of edges, one from each tree, each with bootstrap support strictly greater than $B$, that conflict.
Hence, if $B = 1$, then no two trees can be considered incompatible. 


\begin{thm}[Inconsistency of WSB pipeline followed by reasonable summary method]
	\label{thm:incWSB}
	Under the multi-locus MSC+CFN model, with a single site evolving down each gene tree, the WSB pipeline followed by a reasonable summary method is not statistically consistent.
\end{thm}
\noindent The proof of this theorem is given in the Appendix, and establishes that when each locus has a single site  then there is a $B < 1$ and a tree with topology, branch lengths, and mutation rates 
such that, given data generated under the MSC+CFN model, as the number of loci $m\rightarrow\infty$, the distribution produced by the WSB pipeline with support threshold $B$ is flat.
Hence, the application of $\mathcal{A}$ to this distribution will not converge to the true species tree topology with probability going to $1$.
In other words, the WSB pipeline is not statistically consistent under the MSC+CFN model because uninformative genes
can swamp the bins and produce a flat distribution.

The following modification to the WSB pipeline (which we refer to as the WSB* pipeline) to remove all genes that have no branches with bootstrap support above $B$ addresses
this problem in that the distribution is no longer flat:
\begin{itemize}
	\item Remove all gene trees that do not support any internal edge above the bootstrap threshold $B$ from the analysis before
	doing any binning.
\end{itemize}
However, we still show:
\begin{thm}[Inconsistency of WSB pipeline followed by reasonable summary method]
	\label{thm:incWSB*}
	The WSB* pipeline followed by $\mathcal{A}$ is not only not statistically consistent but is positively misleading. That is, for any length $L\in\mathbb{N}$, there is a $B < 1$ and a species tree with topology, branch lengths and mutation rates such that, given data generated under the MSC+CFN model, as the number of loci $m\rightarrow\infty$,  the topology produced by $\mathcal{A}$ after going through the WSB* pipeline with support threshold $B$ is unique and is different from the true topology with probability going to 1.
\end{thm}


\section*{Theoretical framework}

Our analysis in fact establishes a
stronger---perhaps more counter-intuitive---result. We show that partitioned maximum likelihood, topology-based summary methods, and weighted statistical binning pipelines are statistically inconsistent for multi-locus evolution {\em where there is no gene tree heterogeneity at all}, when all the loci have only $L$ sites for any arbitrarily selected $L$. By a continuity argument, we also establish that these negative results imply that these methods, which were designed to address heterogeneity across the genome resulting from ILS, are also statistically inconsistent under the MSC+CFN model.

\paragraph{Setting for analysis}
Fix $\tree^0$ to be the four-taxon topology $ab|cd$ on $\{a,b,c,d\}$ and let $\Lambda^0$ denote a vector of branch lengths on $\tree^0$ under the CFN model. Specifically, denote the endpoint of the middle edge on the $ab$ side as $e$, and on the $cd$ side as $f$ (see Figure \ref{fig1}). For this tree, denote the length of branch $ae$ as $\lambda^0_a$, $be$ as $\lambda^0_b$, $cf$ as $\lambda^0_c$, $df$ as $\lambda^0_d$ and $ef$ as $\lambda^0_m$.
\begin{figure}
	\begin{center}
		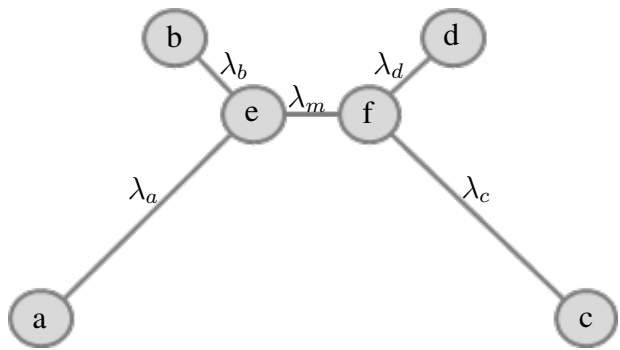
	\end{center}
	\caption{A four-taxon tree}\label{fig1}
\end{figure}
For a branch length $\lambda$, we will also
use the parametrization $\phi 
= -\frac{1}{2} \log \lambda$ in terms of which
the probability of a change along this branch is
\begin{equation}
\label{eq:theta}
p=\frac{1}{2}\left(1-e^{-2\lambda}\right)=\frac{1}{2}(1-\phi),
\end{equation}
and the probability of no change is
$
q=\frac{1}{2}(1+\phi).
$
See \cite[Section 8.6]{semple2003phylogenetics} for more details on this standard parameterization.
Denote the $p$-, $q$-, and $\phi$-parameters as defined above for each branch using the same subscripts.
We choose $\Lambda^0$ to construct a Felsenstein zone tree (i.e., a four-leaf model tree where some tree estimation methods are positively misleading, as shown in \cite{Felsenstein1978}) 
where, for a parameter $\rho > 0$, $p^0_a=p^0_c=\rho$ and $p^0_b=p^0_d=p^0_m=\rho^3$. Note that for any $\rho>0$, we can set $\lambda^0_a=\lambda^0_c=-\frac{1}{2}\log (1-2\rho)$ and $\lambda^0_b=\lambda^0_d=\lambda^0_m=-\frac{1}{2}\log (1-2\rho^3)$ to satisfy this relationship. We assume that
the characters $\chi_{\cdot j}$, $j =1,2,\ldots$, are generated under the CFN model on $(\tree^0,\Lambda^0)$.
We also denote the alternate topologies
by $\tree^*=ac|bd$ and $\tree^1=ad|bc$.

\paragraph{Basic claims}
Our main theorems are implied respectively by the following basic claims.
\begin{claim}[Partitioned ML: Felsenstein zone]
	\label{claim:partML}
	Assume that
	the length-$L$ locus sequences $\chi_{\cdot j}$, $j =1,2,\ldots$, are generated under the CFN model on $(\tree^0,\Lambda^0)$ and let
	$\hat\tree_j$ be the fully partitioned maximum likelihood topology obtained from the sequences of the first $j$ loci. For any length $L \geq 1$, there is $\rho > 0$ small enough such that, with probability one,
	$\hat\tree_j \to \tree^*$ as $j \to +\infty$.
\end{claim}
\begin{claim}[Reasonable summary methods: Felsenstein zone]
	\label{claim:reas}
	Assume that
	the length-$L$ locus sequences $\chi_{\cdot j}$, $j =1,2,\ldots$, are generated under the CFN model on $(\tree^0,\Lambda^0)$ and let
	$\hat\tree_j$ be the topology obtained from a reasonable summary method $\mathcal{A}$ on the sequences of the first $j$ loci using maximum likelihood.
	For any length $L \geq 1$, there is $\rho > 0$ small enough such that, with probability one,
	$\hat\tree_j \to \tree^*$ as $j \to +\infty$.
\end{claim}
\begin{claim}[WSB* pipeline: Felsenstein zone]
	\label{claim:wsb*}
	Let $1 - \frac{2}{3}\left(\frac{1}{L}\right)^L \leq B < 1$.
	Assume that
	the length-$L$ locus sequences $\chi_{\cdot j}$, $j =1,2,\ldots$, are generated under the CFN model on $(\tree^0,\Lambda^0)$ and let
	$\hat\tree_j$ be the topology obtained from the WSB* pipeline with threshold $B$ followed by a reasonable summary method $\mathcal{A}$ on the sequences of the first $j$ loci. There is $\rho > 0$ small enough such that, with probability one,
	$\hat\tree_j \to \tree^*$ as $j \to +\infty$.
\end{claim}
\noindent While the claims above are established under the multi-locus CFN model with a single tree, we show in the Appendix that these results also apply to the MSC+CFN model by choosing a species tree which is highly likely to produce gene trees matching the
species tree. 

\paragraph{Analysis of partitioned ML} 
We describe the main ideas 
used to prove Claim~\ref{claim:partML}.
We proceed as follows:
\begin{enumerate}
	\item[(a)] By choosing $\rho$ small enough,
	we show that we can restrict the analysis
	to the five most common dataset types,
	which we refer to as {\em locus patterns}.
	
	\item[(b)] We then show that, for these locus
	patterns, the likelihood on $\tree^*$ dominates the likelihood on $\tree^0, \tree^1$, and that this domination is strict in one case.
\end{enumerate}
Under our choice of branch lengths, as $\rho \to 0$, the five most common locus patterns, which we refer to as {\em dominant} (see Lemma~\ref{lem:key1} below for justification), are:
\begin{enumerate}
	\item {\it All constant sites:} Every character has the same state on all four taxa, but that state can change from one character to another (e.g. $x^a=x^b=x^c=x^d=0001010$). We let $\patterns_0$ be the set of such datasets and we let $Q_0$ be the probability of observing any $x \in \patterns_0$ under $(\tree^0,\Lambda^0)$.
	
	\item {\it One singleton site on $a$ or $c$:} All sites are constant except for one, on which either $a$ or $c$ is different from all others (e.g. $x^a=0111110$, $x^b=x^c=x^d=1111110$).  We let $\patterns_{11}$ be the set of such datasets and we let $Q_{11}$ be the probability of observing any $x \in \patterns_{11}$ under $(\tree^0,\Lambda^0)$.
	
	\item {\it Two identical singleton sites  on $a$ or $c$:} All sites are constant except for two, each of which has the same taxon $a$ or $c$ different from the others (e.g. $x^a=0011110$, $x^b=x^c=x^d=1111110$).  We let $\patterns_{2=}$ be the set of such datasets and we let $Q_{2=}$ be the probability of observing any $x \in \patterns_{2=}$ under $(\tree^0,\Lambda^0)$.
	
	\item {\it Two different singleton sites  on $a$ and $c$:} All sites are constant except for two, one of which has a different character state on $a$ and the other a different character state on $c$ (e.g. $x^a=1001110$, $x^c=0101110$, $x^b=x^d=0001110$). We let $\patterns_{2\neq}$ be the set of such datasets and we let $Q_{2\neq}$ be the probability of observing any $x \in \patterns_{2\neq}$ under $(\tree^0,\Lambda^0)$.
	
	\item {\it One site with a $2/2$-split $ac|bd$:} $L-1$ sites are constant with a single site having $a$ and $c$ different from $b$ and $d$ (e.g. $x^a=x^c=1001110$, $x^b=x^d=0001110$). We let $\patterns_{12}$ be the set of such datasets and we let $Q_{12}$ be the probability of observing any $x \in \patterns_{12}$ under $(\tree^0,\Lambda^0)$.
\end{enumerate}
\noindent Note that above only the last pattern is informative and it supports the split in $\tree^*$ rather than $\tree^0$.
Let $\widetilde{\patterns}$ be the set of all remaining locus patterns.
\begin{lem}[Dominant patterns and their likelihood contributions]\label{lem:key1}
	\hfill
	\begin{enumerate}
		\item[(a)]
		The probabilities of observing the dominant locus patterns are bounded as follows:
		\begin{equation*}
		Q_0 =\left(\frac{1}{2}\right)^L - \bigo(\rho),
		\quad Q_{11}=\bigo(\rho),
		\quad Q_{2=}=\bigo(\rho^2),
		\end{equation*}
		\begin{equation*}
		\quad Q_{2 \neq}=\bigo(\rho^2) \text{ and }
		Q_{12}=  \left(\frac{1}{2}\right)^L \rho^2 + \bigo(\rho^3).
		\end{equation*}
		Moreover, for all $x \in \widetilde{\patterns}$, the probability of observing $x$ under the CFN model on $(\tree^0,\Lambda^0)$ is $\bigo(\rho^3)$.
		
		\item[(b)]              For all $x \in \patterns_0 \cup \patterns_{11} \cup \patterns_{2=} \cup \patterns_{2\neq}$, it holds that
		\begin{equation*}
		\sup_\Lambda \ell(\tree^*,\Lambda,x) - \sup_\Lambda \ell(\tree^0,\Lambda,x) \geq 0,
		\end{equation*}
		while, for all $x \in \patterns_{12}$,
		\begin{equation*}
		\sup_\Lambda \ell(\tree^*,\Lambda,x) - \sup_\Lambda \ell(\tree^0,\Lambda,x) \geq K_{12} > 0,
		\end{equation*}
		for some positive constant $K_{12}$ depending only
		on $L$. The same holds if one replaces $\tree^0$ with $\tree^1$ above.
	\end{enumerate}
	
\end{lem}
\noindent Note that the big-O notation implicitly includes the contribution from $L$, which we treat as a constant. The detailed proofs of Lemma~\ref{lem:key1} and Claim~\ref{claim:partML} are provided in the Appendix. 
Claim~\ref{claim:reas} follows from a similar argument, which is also detailed in the Appendix. 

\paragraph*{Analysis of WSB* pipeline}
Our analysis of the WSB* pipeline follows
along similar lines. 
Our key additional observation is that, by choosing an appropriate bootstrap
threshold, we ensure that 
the only loci passed on to the summary method are ``saturated,'' that is all their sites correspond to an equivalent character. The rest of the analysis is similar to Claim~\ref{claim:reas} and relies on the fact that the loci passed on to the summary method are dominated by the ``wrong split.''
Formally, we say that two characters are equivalent if they are identical up to switching 0s and 1s.
We say that a locus pattern $x$ is {\it saturated} if all characters in $x$ are equivalent. On four taxa, there are only
three types of saturated patterns:
\begin{enumerate}
	\item {\it All-constant:} Every character has the same value on all four taxa (e.g. $x^a=x^b=x^c=x^d=0001010$). We let $\patsat_0$ be the set of such datasets and we let $\qsat_0$ be the probability of observing any $x \in \patsat_0$ under $(\tree^0,\Lambda^0)$.
	
	\item {\it All-singleton on a fixed taxon:} All sites have the same taxon different from all others (e.g. $x^a=0101111$, $x^b=x^c=x^d=1010000$).  We let $\patsat_{1}$ be the set of such datasets and we let $\qsat_{1}$ be the probability of observing any $x \in \patsat_{1}$ under $(\tree^0,\Lambda^0)$.

	\item {\it All-$2/2$-split with a fixed split:} All sites have two fixed taxa---say, $a$ and $c$---identical while being different from the other two taxa---$b$ and $d$---(e.g. $x^a=x^c=1010111$, $x^b=x^d=0101000$). We let $\patsat_{ac|bd}$ be the set of such datasets 
	for the split $ac|bd$ and we let $\qsat_{ac|bd}$ be the probability of observing any $x \in \patsat_{ac|bd}$ under $(\tree^0,\Lambda^0)$ (and similarly for the other possible splits). For short, we refer to this type of datasets as {\em split-saturated genes}.
\end{enumerate}
\begin{lem}[Saturated genes]\label{lem:key2}
	\hfill
	\begin{enumerate}
		\item[(a)] Under the WSB* pipeline
		with threshold $B  \geq 1 - \frac{2}{3}\left(\frac{1}{L}\right)^L$,
		the only length-$L$ locus sequences passed on to the summary method are the ones in $\patsat_{ac|bd}$,
		$\patsat_{ab|cd}$ and $\patsat_{ad|bc}$. Moreover,
		$$
		\qsat_{ac|bd} =
		\left(\frac{1}{2}\right)^{L} \rho^{2L} + \bigo(\rho^{2L+1}),
		$$
		while
		$$
		\qsat_{ab|cd}
		= \bigo(\rho^{3L}),
		\quad
		\qsat_{ad|bc}
		= \bigo(\rho^{3L}).
		$$
		
		\item[(b)] For any $x \in \patsat_{ab|cd}$,
		the topology $ab|cd$ is the unique ML optimizer. And similarly for the other splits.
	\end{enumerate}
\end{lem}
\noindent The detailed proofs of Lemma~\ref{lem:key2} and Claim~\ref{claim:wsb*} are provided in the Appendix.

\section*{Discussion}

Our results show that fully partitioned maximum
likelihood is inconsistent (even positively misleading)
even when there is no gene tree heterogeneity at all (i.e., when all loci evolve down a common CFN model tree), and hence by continuity 
under the multi-locus MSC+CFN model.
The inconsistency result occurs because each locus has at most $L$ sites (for an arbitrarily selected bound $L$), and the loci all evolve down gene trees that have long branch attraction (LBA).
It is well known that maximum likelihood is statistically consistent even in the presence of LBA, but our results show that LBA is sufficient to bias fully partitioned ML towards the same wrong tree on each locus, and hence towards the same wrong tree for the partitioned concatenation analysis. 

The same argument is used to establish that reasonable summary methods  and weighted statistical binning pipelines that use these reasonable summary methods can be positively misleading when each locus has only $L$ sites, even when there is no gene tree heterogeneity.
Hence, summary methods and weighted statistical binning pipelines do not solve this challenge, either.
All the methods we addressed in this study can be seen as partitioned analyses -- partitioned maximum likelihood estimates numeric parameters for each locus but keeps the tree topology the same across the loci, and summary methods estimate the gene trees independently across the loci. 

The fundamental challenge to multi-locus species tree estimation using these partitioned analyses (whether partitioned maximum likelihood or summary methods) is that  maximum likelihood tree estimation is impacted by conditions such as LBA when the number of sites is not allowed to increase. 

It is interesting to consider unpartitioned maximum likelihood under the same set of conditions.
When all the loci evolve down the same CFN model tree,  even though each locus has only $L$ sites, as the number of loci increases, the unpartitioned maximum likelihood analysis will converge to the true tree; thus, 
unpartitioned maximum likelihood analysis is consistent under this setting.
On the other hand, when there is gene tree heterogeneity resulting from ILS (as modelled by the MSC), then unpartitioned ML is inconsistent and can be positively misleading \cite{RochSteel2015}.
Hence, unpartitioned maximum likelihood can be statistically consistent under one setting and inconsistent (and even positively misleading) under another. 
In other words, unpartitioned maximum likelihood is not the solution to the challenge raised by this study.


Our analysis does not apply to multilocus methods that estimate the species tree directly from 
sequence data---without a gene tree reconstruction step. These include for instance METAL~\cite{Dasarathy2015}, SNAPP~\cite{bryant2012inferring},
SVDquartets~\cite{ChifmanKubatko:14,chifman2015identifiability}, and *BEAST \cite{heled2009bayesian}.
In particular, METAL has been shown to be consistent on finite-length genes under some assumptions on the multispecies coalescent~\cite{Dasarathy2015}.
It is also worthwhile pointing out that our results, while being based on the MSC, are likely to hold more generally for other sources of gene tree discordance, including horizontal gene transfer (HGT). Indeed, as long as rates of HGT are low enough, in the Felsenstein zone similar conclusions about inconsistency will follow for partitioned ML and summary-based methods.

\section*{Conclusion}
Prior to this study, many
coalescent-based species tree estimation methods were assumed to
be statistically consistent under this regime, but no
proofs had been provided.
This study now establishes that all the
standard methods used in phylogenomic species tree estimation
are statistically inconsistent.

Moreover, only a very small number of methods have been proven
to be statistically consistent for bounded $L$.
Some of the summary methods described
in \cite{RochWarnow2015} are statistically consistent
for $L=1$, but the proofs depend on the strict molecular clock.
Similarly, SVDquartets \cite{SVDquartets} (a site-based method for estimating quartet trees from a single site per locus) is based on an identifiability result that depends on the strict molecular clock; however, the species tree estimation method itself has not yet established to be statistically consistent
under the MSC even when the strict molecular clock holds.

Furthermore, when the strict molecular clock
assumption does not hold,  very few methods
are statistically consistent for bounded $L$.
METAL \cite{Dasarathy2015} is one of the few coalescent-based methods that does not require a molecular clock, and that has been proven to be statistically consistent under the MSC+CFN model. It should be noted however that the model of evolution in \cite{Dasarathy2015} allows mutation rates to vary across branches of the species tree, but those rates must be the same across loci, a major constraint.
Much remains to be understood about the important theoretical question of fixed locus length consistency of multilocus method in general.

\section*{Acknowledgments}

This work was supported by funding from the U.S. National Science Foundation CCF-1535977 (to TW).
SR was supported by NSF grants DMS-1149312 (CAREER), DMS-1614242 and CCF-1740707 (TRIPODS).

\bibliographystyle{plain}
\bibliography{nute-roch-warnow}

\begin{thebibliography}{10}

\bibitem{Allman}
E.S. Allman, J.H. Degnan, and J.A. Rhodes.
\newblock Identifying the rooted species tree from the distribution of unrooted
  gene trees under the coalescent.
\newblock {\em Journal of Mathematical Biology}, 62:833--862, 2011.

\bibitem{WSB}
Md.~S. Bayzid, S.~Mirarab, B.~Boussau, and T.~Warnow.
\newblock Weighted statistical binning: enabling statistically consistent
  genome-scale phylogenetic analyses.
\newblock {\em PLOS One}, 2015.
\newblock DOI: 10.1371/journal.pone.0129183.

\bibitem{naive-binning}
Md~Shamsuzzoha Bayzid and Tandy Warnow.
\newblock {Naive binning improves phylogenomic analyses}.
\newblock {\em Bioinformatics}, 29(18):2277--84, 2013.

\bibitem{bryant2012inferring}
David Bryant, Remco Bouckaert, Joseph Felsenstein, Noah~A Rosenberg, and
  Arindam RoyChoudhury.
\newblock Inferring species trees directly from biallelic genetic markers:
  bypassing gene trees in a full coalescent analysis.
\newblock {\em Molecular biology and evolution}, 29(8):1917--1932, 2012.

\bibitem{Cannon2016}
Johanna~Taylor Cannon, Bruno~Cossermelli Vellutini, Julian Smith, Fredrik
  Ronquist, Ulf Jondelius, and Andreas Hejnol.
\newblock Xenacoelomorpha is the sister group to nephrozoa.
\newblock {\em Nature}, 530(7588):89--93, 2016.

\bibitem{Cavender1978}
J.~A. Cavender.
\newblock Taxonomy with confidence.
\newblock {\em Math. Biosci.}, 40:271–280, 1978.

\bibitem{SVDquartets}
J.~Chifman and L.~Kubatko.
\newblock Quartet inference from {SNP} data under the coalescent.
\newblock {\em Bioinformatics}, 30(23):3317--3324, 2014.

\bibitem{ChifmanKubatko:14}
Julia Chifman and Laura Kubatko.
\newblock Quartet inference from snp data under the coalescent model.
\newblock {\em Bioinformatics}, 30(23):3317--3324, 2014.

\bibitem{chifman2015identifiability}
Julia Chifman and Laura Kubatko.
\newblock Identifiability of the unrooted species tree topology under the
  coalescent model with time-reversible substitution processes, site-specific
  rate variation, and invariable sites.
\newblock {\em Journal of theoretical biology}, 374:35--47, 2015.

\bibitem{chou-recombcg}
J.~Chou, A.~Gupta, S.~Yaduvanshi, R.~Davidson, M.~Nute, S.~Mirarab, and
  T.~Warnow.
\newblock A comparative study of {SVDquartets} and other coalescent-based
  species tree estimation methods.
\newblock {\em BMC Genomics}, 16(Suppl 10):S2, 2015.

\bibitem{Dasarathy2015}
Gautam Dasarathy, Robert Nowak, and Sebastien Roch.
\newblock Data requirement for phylogenetic inference from multiple loci: A new
  distance method.
\newblock {\em IEEE/ACM Trans. Comput. Biol. Bioinformatics}, 12(2):422--432,
  March 2015.

\bibitem{Durrett:96}
Richard Durrett.
\newblock {\em Probability: theory and examples}.
\newblock Duxbury Press, Belmont, CA, second edition, 1996.

\bibitem{Farris1973}
J.~S. Farris.
\newblock A probability model for inferring evolutionary trees.
\newblock {\em Syst. Zool.}, 22:250– 256, 1973.

\bibitem{Felsenstein1978}
J~Felsenstein.
\newblock {Cases in which Parsimony or Compatibility Methods Will be Positively
  Misleading}.
\newblock {\em Systematic Zoology}, 27(4):401--410, December 1978.

\bibitem{heled2009bayesian}
Joseph Heled and Alexei~J Drummond.
\newblock Bayesian inference of species trees from multilocus data.
\newblock {\em Molecular biology and evolution}, 27(3):570--580, 2009.

\bibitem{jarvis-2014b}
E.~Jarvis, S.~Mirarab, A.~J. Aberer, B.~Li, P.~Houde, C.~Li, S.Y.W.~Ho 
  Faircloth, B.~Nabholz, J.~T. Howard, A.~Suh 
  J.~Li, F.~Zhang, 
  Boussau, Md.~S. Bayzid, V.~Zavidovych, S.~Subramanian 
  S.~Capella-Guti\'errez, J.~Huerta-Cepas 
  M.~Schierup, B.~Lindow 
  X.~Zhan, A.~Dixon, S.~Li, N.~Li, Y.~Huang, 
  Bertelsen, F.~H. Sheldon, 
  M.~Wirthlin 
  A.~M.~Vargas Velazquez, A.~Alfaro-N\'unez, P.~F.~Campos 
  T.~Sicheritz-Ponten, A.~Pas, T.~Bailey 
  Lambert, Q.~Zhou, 
  Y.~Zeng, S.~Liu, Z.~Li, B.~Liu, K.~Wu, 
  Y.~Zhang, H.~Yang 
  S.~O’Brien, D.~Haussler, O.~A. Ryder, 
  Graves, T.~C.~Glenn 
  T.~Warnow, W.~Jun, M.~T.~P. Gilbert, and G.~Zhang.
\newblock Whole-genome analyses resolve early branches in the tree of life of
  modern birds.
\newblock {\em Science}, 346(6215):1320--1331, 2014.

\bibitem{Jukes1969}
T.~H. Jukes and C.~R. Cantor.
\newblock Evolution of protein molecules.
\newblock {\em Mammalian Protein Metabolism}, pages 21--132, 1969.

\bibitem{stem}
L~S Kubatko, B~C Carstens, and L~L Knowles.
\newblock {STEM}: species tree estimation using maximum likelihood for gene
  trees under coalescence.
\newblock {\em Bioinformatics}, 25(7):971--973, 2009.

\bibitem{leache2010accuracy}
Adam~D Leach{\'e} and Bruce Rannala.
\newblock The accuracy of species tree estimation under simulation: a
  comparison of methods.
\newblock {\em Systematic biology}, 60(2):126--137, 2010.

\bibitem{njst}
L.~Liu and L.~Yu.
\newblock Estimating species trees from unrooted gene trees.
\newblock {\em Syst. Biol.}, 60(5):661--667, 2011.

\bibitem{Liu2010a}
Liang Liu, Lili Yu, and Scott~V Edwards.
\newblock {A maximum pseudo-likelihood approach for estimating species trees
  under the coalescent model}.
\newblock {\em BMC Evolutionary Biology}, 10(1):302, 2010.

\bibitem{star}
Liang Liu, Lili Yu, Dennis~K Pearl, and Scott~V Edwards.
\newblock Estimating species phylogenies using coalescence times among
  sequences.
\newblock {\em Systematic Biology}, 58(5):468--477, 2009.

\bibitem{DavidMaddison2016}
D.~Maddison.
\newblock The rapidly changing landscape of insect phylogenetics.
\newblock {\em Current Opinion in Insect Science}, 18:77--82, 2016.

\bibitem{Maddison1997a}
W.~P. Maddison.
\newblock Gene trees in species trees.
\newblock {\em Systematic Biology}, 46(3):523--536, September 1997.

\bibitem{Meredith2011}
Robert~W Meredith, Jan~E Jane\v{c}ka, John Gatesy, Oliver~a Ryder, Colleen~a
  Fisher, Emma~C Teeling, Alisha Goodbla, Eduardo Eizirik, Taiz L~L Sim\~{a}o,
  Tanja Stadler, Daniel~L Rabosky, Rodney~L Honeycutt, John~J Flynn, Colleen~M
  Ingram, Cynthia Steiner, Tiffani~L Williams, Terence~J Robinson, Angela
  Burk-Herrick, Michael Westerman, Nadia~a Ayoub, Mark~S Springer, and
  William~J Murphy.
\newblock {Impacts of the Cretaceous Terrestrial Revolution and KPg extinction
  on mammal diversification.}
\newblock {\em Science (New York, N.Y.)}, 334(6055):521--4, 2011.

\bibitem{astral}
S.~Mirarab, R.~Reaz, Md.~S. Bayzid, T.~Zimmermann, M.S. Swenson, and T.~Warnow.
\newblock {ASTRAL}: {A}ccurate {S}pecies {TR}ee {AL}gorithm.
\newblock {\em Bioinformatics}, 30(17):i541--i548, 2014.

\bibitem{astral-2}
S.~Mirarab and T.~Warnow.
\newblock {ASTRAL-II}: coalescent-based species tree estimation with many
  hundreds of taxa and thousands of genes.
\newblock {\em Bioinformatics}, 31(12):i44--i52, 2015.

\bibitem{Mirarab2014}
Siavash Mirarab, Md.~Shamsuzzoha Bayzid, Bastien Boussau, and Tandy Warnow.
\newblock Statistical binning improves species tree estimation in the presence
  of gene tree incongruence.
\newblock {\em Science}, 2014.
\newblock (Under review, companion paper for Avian Phylogenomics project).

\bibitem{Misof-2014}
Bernhard Misof, Shanlin Liu, Karen Meusemann, Ralph~S. Peters, Alexander
  Donath, Christoph Mayer, Paul~B. Frandsen, Jessica Ware, Tom{\'a}{\v s}
  Flouri, Rolf~G. Beutel, Oliver Niehuis, Malte Petersen, Fernando
  Izquierdo-Carrasco, Torsten Wappler, Jes Rust, Andre~J. Aberer, Ulrike
  Asp{\"o}ck, Horst Asp{\"o}ck, Daniela Bartel, Alexander Blanke, Simon Berger,
  Alexander B{\"o}hm, Thomas~R. Buckley, Brett Calcott, Junqing Chen, Frank
  Friedrich, Makiko Fukui, Mari Fujita, Carola Greve, Peter Grobe, Shengchang
  Gu, Ying Huang, Lars~S. Jermiin, Akito~Y. Kawahara, Lars Krogmann, Martin
  Kubiak, Robert Lanfear, Harald Letsch, Yiyuan Li, Zhenyu Li, Jiguang Li,
  Haorong Lu, Ryuichiro Machida, Yuta Mashimo, Pashalia Kapli, Duane~D.
  McKenna, Guanliang Meng, Yasutaka Nakagaki, Jos{\'e}~Luis Navarrete-Heredia,
  Michael Ott, Yanxiang Ou, G{\"u}nther Pass, Lars Podsiadlowski, Hans Pohl,
  Bj{\"o}rn~M. von Reumont, Kai Sch{\"u}tte, Kaoru Sekiya, Shota Shimizu, Adam
  Slipinski, Alexandros Stamatakis, Wenhui Song, Xu~Su, Nikolaus~U. Szucsich,
  Meihua Tan, Xuemei Tan, Min Tang, Jingbo Tang, Gerald Timelthaler, Shigekazu
  Tomizuka, Michelle Trautwein, Xiaoli Tong, Toshiki Uchifune, Manfred~G.
  Walzl, Brian~M. Wiegmann, Jeanne Wilbrandt, Benjamin Wipfler, Thomas K.~F.
  Wong, Qiong Wu, Gengxiong Wu, Yinlong Xie, Shenzhou Yang, Qing Yang, David~K.
  Yeates, Kazunori Yoshizawa, Qing Zhang, Rui Zhang, Wenwei Zhang, Yunhui
  Zhang, Jing Zhao, Chengran Zhou, Lili Zhou, Tanja Ziesmann, Shijie Zou,
  Yingrui Li, Xun Xu, Yong Zhang, Huanming Yang, Jian Wang, Jun Wang, Karl~M.
  Kjer, and Xin Zhou.
\newblock Phylogenomics resolves the timing and pattern of insect evolution.
\newblock {\em Science}, 346(6210):763--767, 2014.

\bibitem{molloy2017}
Erin~K. Molloy and Tandy Warnow.
\newblock To include or not to include: The impact of gene filtering on species
  tree estimation methods.
\newblock {\em Systematic Biology}, page syx077, 2017.

\bibitem{Mossel2010}
Elchanan Mossel and Sebastien Roch.
\newblock Incomplete lineage sorting: consistent phylogeny estimation from
  multiple loci.
\newblock {\em IEEE/ACM Transactions on Computational Biology and
  Bioinformatics}, 7(1):166--71, January 2010.

\bibitem{Neyman1971}
J.~Neyman.
\newblock Molecular studies of evolution: a source of novel statistical
  problems.
\newblock In S.S Gupta and J.~Yackel, editors, {\em Statistical decision theory
  and related topics}, page 1–27. Academic Press, New York and London, 1971.

\bibitem{patel2013}
Swati Patel, Rebecca~T Kimball, and Edward~L Braun.
\newblock Error in phylogenetic estimation for bushes in the tree of life.
\newblock {\em Journal of Phylogenetics and Evolutionary Biology}, 1(2):110,
  2013.

\bibitem{Posada2016}
David Posada.
\newblock Phylogenomics for systematic biology.
\newblock {\em Systematic Biology}, 65(3):353--356, 2016.

\bibitem{RochSteel2015}
S.~Roch and M.A. Steel.
\newblock Likelihood-based tree reconstruction on a concatenation of aligned
  sequence data sets can be statistically inconsistent.
\newblock {\em Theor. Popul. Biol.}, 100:56--62, 2015.

\bibitem{RochWarnow2015}
S.~Roch and T.~Warnow.
\newblock On the robustness to gene tree estimation error (or lack thereof) of
  coalescent-based species tree methods.
\newblock {\em Systematic Biology}, 64(4):663--676, 2015.

\bibitem{semple2003phylogenetics}
Charles Semple and Mike~A Steel.
\newblock {\em Phylogenetics}, volume~24.
\newblock Oxford University Press on Demand, 2003.

\bibitem{tavare-gtr}
S.~Tavar\'e.
\newblock Some probabilistic and statistical problems in the analysis of {DNA}
  sequences.
\newblock In {\em Lectures on Mathematics in the Life Sciences}, volume~17,
  pages 57--86. American Mathematical Society, 1986.

\bibitem{astrid}
P.~Vachaspati and T.~Warnow.
\newblock { ASTRID}: {Accurate Species TRees from Internode Distances}.
\newblock {\em BMC Genomics}, 16(Suppl 10):S3, 2015.

\bibitem{Warnow-PC2015}
T.~Warnow.
\newblock Concatenation analyses in the presence of incomplete lineage sorting.
\newblock {\em PLOS Currents: Tree of Life}, 2015.
\newblock doi: 10.1371/currents.tol.8d41ac0f13d1abedf4c4a59f5d17b1f7.

\bibitem{wickett2014}
N.J. Wickett, S.~Mirarab, N.~Nguyen, T.~Warnow, E.~Carpenter, N.~Matasci,
  S.~Ayyampalayam, M.S. Barker, J.G. Burleigh, M.A. Gitzendanner, , B.~R.
  Ruhfel, E.~Wafula, J.~P. Der, S.~W. Graham, S.~Mathews 
  Soltis, P.~S. Soltis, N.~W.~Miles 
  L.~DeGironimo 
  H.~Philippe, C.~W. dePamphilis, T.~Chen, M.~K.~Deyholos 
  T.~M. Kutchan, M.~M. Augustin, J.~Wang 
  X.~Sun 
\newblock Phylotranscriptomic analysis of the origin and early diversification
  of land plants.
\newblock {\em Proceedings of the National Academy of Sciences},
  111(45):E4859--E4868, 2014.

\end{thebibliography}

\appendix

\section{Inconsistency of WSB on a single site}
\label{sec:app-wsb}

Here we show that the weighted statistical binning pipeline
(as defined in \cite{WSB})
is {\em inconsistent} for any   positive $L$, for
some four-species model tree. 
We begin with a lemma.
\begin{lem}
	Let $\stree$ be a model species tree with four species $a,b,c,d$, and
	suppose every locus has only one site.
	In a weighted statistical binning pipeline with bootstrap
	support threshold $B \geq \frac{1}{3}$, 
	there will at most three bins (one for each of the three
	possible binary topologies on four leaves), and
	the bin associated with topology $ab|cd$ will have
	all the ML-informative genes that support $ab|cd$.
	\label{lemma1}
\end{lem}
\begin{proof}
	Because there is only one site for each gene, the ML-informative
	genes have bootstrap support of 100\%. Hence, no
	two ML-informative genes can be placed in the same bin if
	they support different tree topologies. Therefore,
	for any bin, 
	the ML-informative genes placed in the bin will support the
	same topology.
	Also, the ML-uninformative genes produce trees with bootstrap
	support equal to $\frac{1}{3}$, since every tree topology has equal 
	maximum likelihood score. 
	These genes are therefore considered compatible with every other gene,
	since the bootstrap support threshold $B \geq \frac{1}{3}$.
	
	Since there are only three tree topologies,  the incompatibility graph 
	is the union of a complete 3-partite graph (defined by the ML-informative
	genes) and a collection of isolated vertices (defined by the ML-uninformative
	genes). 
	Hence, the incompatibility graph
	can be $3$-colored. 
	Since statistical binning seeks the minimum vertex coloring
	for the incompatibility graph, it will partition the
	genes into three bins, with one bin for each binary tree topology. 
	Hence, the ML-informative genes are partitioned into three sets
	based on the tree topology they support.
\end{proof}

We continue with an analysis of WSB pipelines followed
by reasonable summary methods, beginning with the case of
a single site per gene. The following
result implies Theorem~\ref{thm:incWSB}.
\begin{thm}
	Suppose every gene has only one site, and
	let $(\stree,\Gamma,\theta)$ be a MSC+CFN model species tree
	with leaves $a,b,c,d$.
	Let $B \geq \frac{1}{3}$. 
	If for all binary trees
	$t$ on $a,b,c,d$ the probability that a random
	gene is ML-informative and supports $t$
	is at most $\frac{1}{3}$, then 
	weighted statistical binning followed
	by a reasonable summary method will be statistically
	inconsistent.
	\label{thm:wsb-one-site}
\end{thm}
\begin{proof}
	The argument will establish that 
	under the conditions of the theorem, as
	the number of genes increases, the WSB binning
	process will converge to a flat 
	distribution on the three possible
	tree topologies on $a,b,c,d$, so that
	any reasonable summary method will be
	inconsistent.

	By Lemma \ref{lemma1},
	in a weighted statistical binning
	pipeline, there will be three bins (one for
	each binary tree topology), and
	the bin for binary tree $t$ will have all genes that
	are ML-informative and
	support the split for $t$, and may also have
	ML-uninformative genes.
	Furthermore, the ML-uninformative genes can be distributed to
	the bins arbitrarily, since their bootstrap support is
	exactly $\frac{1}{3}$ and $B \geq \frac{1}{3}$.

	Since every gene has only one site, the supergene alignment
	associated to the bin for $ab|cd$ will consist of sites
	that all split $ab|cd$. Hence, 
	when a fully partitioned ML analysis is applied to
	the bin for $t$, the resultant supergene tree will be
	the tree $t$. In a WSB pipeline, the supergene trees for
	each bin will be replicated as many times as the number of
	genes in the bin for $t$.
	These trees are the newly computed gene trees that will
	be passed to the reasonable summary method.
	
	The division of genes into bins attempts to achieve balanced bins,
	so that the number of genes in each bin should be as close to the same as 
	possible. Therefore, if the probability that a gene is
	ML-uninformative
	is sufficiently high, then it will be possible to achieve balanced bins,
	and the distribution of newly computed
	gene trees  will converge to the 
	flat distribution.
	Since reasonable summary methods cannot infer the species tree
	from flat distributions, this means that when
	the probability of being ML-uninformative is sufficiently high, then
	WSB pipelines based on reasonable summary methods will not be
	statistically
	consistent.
\end{proof}

\section{Proofs of the main results}

We provide detailed proofs of the main claims.

\subsection*{Key lemmas}

\begin{proof}[Proof of Lemma~\ref{lem:key1}]
	(a) 	Under our choice of branch lengths, as $\mu \to 0$, the five most common locus site patterns are:
	\begin{enumerate}
		\item {\it All constant sites:} Every character has the same value on all four taxa (e.g. $x^a=x^b=x^c=x^d=000101$). For any such $x \in \patterns_0$, $x$ occurs with probability
		\begin{equation*}
		Q_0 
		= \left[\frac{1}{2}\left(1-\rho^3 \right)^3 \left(1-\rho \right)^2 +\mathcal{O}(\rho)\right]^L
		=\left(\frac{1}{2}\right)^L - \bigo(\rho),
		\end{equation*}
		where the first term in the brackets corresponds
		to the case of no substitution, while the second term accounts for all possibilities with at least
		one substitution. For convenience we denote the expression in brackets---the probability of a single site being identical on all four taxa---as $q_0$. 
		
		\item {\it One singleton site on $a$ or $c$:} All sites are constant except for one, on which either $a$ or $c$ is different from all others (e.g. $x^a=01\dots$, $x^b=x^c=x^d=11\dots$).  Any dataset with this locus site pattern occurs with probability
		\begin{equation*}
		Q_{11}=q_0^{L-1}\left[\frac{1}{2}(1-\rho^3 )^3 (1-\rho)\rho + \bigo(\rho^2)\right]=\bigo(\rho),
		\end{equation*}
		where the first term in the brackets corresponds 
		to the case of a single substitution along the
		edge leading to the differing taxon, while the second
		term accounts for all possibilities involving
		at least two substitutions.
		
		\item {\it Two identical singleton sites  on $a$ or $c$:} All sites are constant except for two, each of which has the same taxon $a$ or $c$ different from the others (e.g. $x^a=001\dots$, $x^b=x^c=x^d=111\dots$).  Any dataset with this locus site pattern occurs with probability
		\begin{equation*}
		Q_{2=}=q_0^{L-2}\left[\frac{1}{2}(1-\rho^3 )^3 (1-\rho)\rho + \bigo(\rho^2) \right]^2=\bigo(\rho^2),
		\end{equation*}
		which follows from the same computation as
		in the one singleton case.
		
		\item {\it Two different singleton sites  on $a$ and $c$:} All sites are constant except for two, one of which has a different character on $a$ and the other a different character on $c$ (e.g. $x^a=100\dots$, $x^c=010\dots$, $x^b=x^d=000\dots$). Any dataset with this locus site pattern occurs with probability
		\begin{equation*}
		Q_{2 \neq}=q_0^{L-2}\left[\frac{1}{2}(1-\rho^3 )^3 (1-\rho)\rho + \bigo(\rho^2) \right]^2=\bigo(\rho^2),
		\end{equation*}
		which follows from the same computation as
		in the one singleton case.
		
		\item {\it One site with a $2/2$-split $ac|bd$:} $L-1$ sites are constant with a single site having $a$ and $c$ different from $b$ and $d$ (e.g. $x^a=x^c=100\dots$, $x^b=x^d=000\dots$). Any dataset with this locus site pattern occurs with probability
		\begin{equation}
		\label{eq:q12}
		Q_{12}=q_0^{L-1}\left[\frac{1}{2}(1-\rho^3 )^3\rho^2
		+ \bigo(\rho^3) \right] =  \left(\frac{1}{2}\right)^L \rho^2 + \bigo(\rho^3),
		\end{equation}
		where the first term in the brackets corresponds 
		to the case of substitutions along the
		edges leading to the differing taxa, while the second
		term accounts for all possibilities with at least 
		one substitution along the other edges.
	\end{enumerate}
	
	Any remaining locus site pattern must include either a change along one of the short branches, which involves multiplication by $\rho^3$, or three changes along one of the long branches, which also means multiplication by $\rho^3$. Thus all $x$ in $\widetilde{\patterns}$ have probability $\mathcal{O}(\rho^3)$. That concludes the proof of the claim in (a).
	
	\medskip
	
	(b) 	It remains to prove (b).
	For each locus site pattern we will put an upper bound on the maximum of the likelihood function for topology $\tree^0=ab|cd$, and show that in every case the alternate topology $\tree^*=ac|bd$ has maximum likelihood greater than or equal to this upper bound, and in at least one case is strictly greater.

	Some remarks about notation first. Note that the labels we have used for the branch lengths of $\tree^0$ can be used similarly regardless of the topology of the tree: $\lambda_m$ represents the middle branch in any topology, and the others represent the branch leading to their respective taxon. Also we use $\Lambda$ and $\Phi$ interchangeably,
	where $\Phi$ is the corresponding collection of $\phi$-parameters as defined in~\eqref{eq:theta}.
	Finally we will use the following property of the $\phi$-parametrization~\cite{semple2003phylogenetics}: the $\phi$'s multiply along paths; indeed, we have for instance,
	\begin{align}
	&\P_{x \sim (\tree^0,\Phi^0)}[x^a_1 \neq x^b_1]\nonumber\\
	&= (1-p_a^0)p_b^0 + p_a^0(1- p_b^0)\nonumber\\
	&= \frac{1}{2}(1+\phi_a^0)\frac{1}{2}(1-\phi_b^0)
	+\frac{1}{2}(1-\phi_a^0)\frac{1}{2}(1+\phi_b^0)\nonumber\\
	&= \frac{1}{2}(1-\phi_a^0\phi_b^0).\label{eq:thetaMultiply}
	\end{align}
	
	Finally, because by inclusion the probability
	of observing $\chi_{\cdot 1}$ is at most the probability
	of observing $\chi_{a 1}$, which is simply $\left(\frac{1}{2}\right)^L$ by independence of the sites, we have
	\begin{equation}
	\label{logLLowerbound2-key}
	\sup_{\Lambda} \ell(\tree,\Lambda,\chi_{\cdot 1})
	\leq \log \left(\frac{1}{2}\right)^{L}
	= -L \log 2.
	\end{equation}
	
	We divide up the proof of by locus site pattern.
	\begin{enumerate}
		\item {\it All constant sites:} Recall from~\eqref{logLLowerbound2-key} that, for any $\tree$ (and, in particular, for $\tree^0$),
		$$\sup_\Lambda \ell(\tree,\Lambda,x) \leq -L\log 2.$$ For $x \in \patterns_0$, that can always (in particular, for $\tree^*$) be achieved by setting all branch lengths to $0$. 
		
		\item {\it One singleton site  on $a$ or $c$:} Without loss of generality, assume the non-constant site is site 1 and that it has $(x^a_1,x^b_1,x^c_1,x^d_1)=(1,0,0,0)$. Assume also
		that $(x^a_i,x^b_i,x^c_i,x^d_i)=(0,0,0,0)$ for all $i = 2,\ldots,L$. We can put the following upper bound on the likelihood function for $\tree^0$.
		Letting $\phi_{ab}=\ta\tb$
		and using~\eqref{eq:thetaMultiply}, we have
		\begin{eqnarray}
		&&\P_{x \sim (\tree^0,\Phi)}\left(x^a_1{=}1,x^b_1{=}x^c_1{=}x^d_1{=}0\right)\nonumber\\
		&&\qquad \times\P_{x \sim (\tree^0,\Phi)}\left(x^a_i{=}x^b_i{=}x^c_i{=}x^d_i{=}0\right)^{L-1}\nonumber\\
		&& \le \P_{x \sim (\tree^0,\Phi)}\left(x^a_1{=}1{\ne}x^b_1\right)\,\P_{x \sim (\tree^0,\Phi)}\left(x^a_i{=}0{=}x^b_i\right)^{L-1}\nonumber\\
		&& \le \frac{1}{2}\left(\frac{1-\phi_{ab}}{2}\right)\left[\frac{1}{2}\left(\frac{1+\phi_{ab}}{2}\right)\right]^{L-1},\label{eq:pattern2Upper}
		\end{eqnarray}
		where the first inequality follows by inclusion.
		To derive our upper bound, we maximize the expression on the last line as a function of $\phi_{ab}$. Taking the log, differentiating and equating to $0$,
		we get
		\begin{equation*}
		\frac{-1}{1-\phi_{ab}}+(L-1)\frac{1}{1+\phi_{ab}}=0
		\end{equation*}
		that is, 
		$\phi_{ab}=\frac{L-2}{L}$.
		Plugging this back above, we get the upper bound 
		\begin{align*}
		&\sup_\Phi \ell(\tree^0,\Phi,x)\nonumber\\
		&\le L \log\left(\frac{1}{2}\right) + \log \left(\frac{1}{L}\right) + (L-1) \log  \left(1-\frac{1}{L}\right).
		\end{align*}
		On the other hand, for $\tree^*$ (or, in fact, any topology), setting $\lambda_b=\lambda_c=\lambda_d=\lambda_m=0$ and $\lambda_a$ so that $p_a=\frac{1}{L}$, we get the matching bound
		\begin{eqnarray*}
			&&\P_{x \sim (\tree^*,\Phi)}\left(x^a_1{=}1,x^b_1{=}x^c_1{=}x^d_1{=}0\right)\nonumber\\
			&& \qquad\times\P_{x \sim (\tree^*,\Phi)}\left(x^a_i{=}x^b_i{=}x^c_i{=}x^d_i{=}0\right)^{L-1}\\
			&& = \frac{1}{2}\left(\frac{1}{L}\right)\left[\frac{1}{2}\left(1 - \frac{1}{L}\right)\right]^{L-1},
		\end{eqnarray*}
		which establishes the required lower bound
		on $\sup_\Phi \ell(\tree^*,\Phi,x)$.
		
		\item {\it Two identical singleton sites  on $a$ or $c$:}
		For this locus site pattern, the argument is identical to the previous locus site pattern, with the difference that the exponents in~\eqref{eq:pattern2Upper} are $2$ and $L-2$, and accordingly throughout, giving an optimal $\phi_{ab}$ of $\frac{L-4}{L}$ and the upper bound $L \log (1/2) + 2\log \left(\frac{2}{L}\right) + (L-2) \log \left(1-\frac{2}{L}\right)$. This can likewise be achieved with topology $\tree^*$ (or, in fact, any topology) if $\lambda_b=\lambda_c=\lambda_d=\lambda_m=0$ and $\lambda_a$ is set so that $p_a=\frac{2}{L}$. 
		
		\item {\it Two different singleton sites  on $a$ and $c$:}
		Assume that $(x^a_1,x^b_1,x^c_1,x^d_1)=(1,0,0,0)$, 
		$(x^a_1,x^b_1,x^c_1,x^d_1)=(0,0,1,0)$ and $(x^a_i,x^b_i,x^c_i,x^d_i)=(0,0,0,0)$ for all $i = 3,\ldots,L$, without loss of generality. (Recall that the case of two different singletons not involving $a$ and $c$ has negligible probability of being observed by part (a) and is therefore not considered here.) We will use the following
		property of the CFN model: on $\tree^0$, because
		the path joining $a,b$ and the path joining $c,d$ are disjoint, the event $\{x^c_1 = x^d_1\}$ is independent of the states $x^a_1$ and $x^b_1$. This is immediate by the symmetry of the CFN model and the Markov property~\cite{semple2003phylogenetics}. (Indeed, conditioning on the state at $f$ has no effect on the agreement between $c$ and $d$.) Using this fact as well as inclusion and~\eqref{eq:thetaMultiply}, we get
		\begin{eqnarray*}
			&&\P_{x \sim (\tree^0,\Phi)}\left(x^a_1{=}1,x^b_1{=}x^c_1{=}x^d_1{=}0\right)\\
			&&\qquad \times \P_{x \sim (\tree^0,\Phi)}\left(x^c_1{=}1,x^a_1{=}x^b_1{=}x^d_1{=}0\right)\\
			&& \qquad \times\P_{x \sim (\tree^0,\Phi)}\left(x^a_i{=}x^b_i{=}x^c_i{=}x^d_i{=}0\right)^{L-2}\\
			&& \le \P_{x \sim (\tree^0,\Phi)}\left(x^a_1{=}1{\ne}x^b_1,x^c_1{=}x^d_1\right)\\
			&& \qquad \times \P_{x \sim (\tree^0,\Phi)}\left(x^a_1{=}0{=}x^b_1,x^c_1{\ne}x^d_1\right)\\ 
			&& \qquad \times \P_{x \sim (\tree^0,\Phi)}\left(x^a_i{=}0{=}x^b_i,x^c_i{=}x^d_i\right)^{L-2}\\
			&& = [\P_{x \sim (\tree^0,\Phi)}\left(x^a_1{=}1{\ne}x^b_1\right)\,\P_{x \sim (\tree^0,\Phi)}\left(x^c_1{=}x^d_1\right)]\\
			&& \qquad \times [\P_{x \sim (\tree^0,\Phi)}\left(x^a_1{=}0{=}x^b_1\right)\,\P_{x \sim (\tree^0,\Phi)}\left(x^c_1{\ne}x^d_1\right)]\\ 
			&& \qquad \times [\P_{x \sim (\tree^0,\Phi)}\left(x^a_i{=}0{=}x^b_i\right)
			\, \P_{x \sim (\tree^0,\Phi)}\left(x^c_i{=}x^d_i\right)]^{L-2}\\
			&& = \frac{1}{2}\left(\frac{1-\phi_{ab}}{2}\right)
			\left(\frac{1+\phi_{cd}}{2}\right)\\
			&& \qquad \qquad  \times \frac{1}{2}\left(\frac{1+\phi_{ab}}{2}\right)
			\left(\frac{1-\phi_{cd}}{2}\right)\\
			&& \qquad \times \left[\frac{1}{2}\left(\frac{1+\phi_{ab}}{2}\right)\left(\frac{1+\phi_{cd}}{2}\right)\right]^{L-2}\\
			&& = \left(\frac{1}{2}\right)^L
			\left(\frac{1-\phi_{ab}}{2}\right)
			\left(\frac{1+\phi_{ab}}{2}\right)^{L-1}\nonumber\\
			&&\qquad \times \left(\frac{1-\phi_{cd}}{2}\right)
			\left(\frac{1+\phi_{cd}}{2}\right)^{L-1},
		\end{eqnarray*}	
		where $\phi_{ab} = \phi_a \phi_b$ and
		$\phi_{cd} = \phi_c \phi_d$. Maximizing
		this last expression over $\phi_{ab}$ and
		$\phi_{cd}$ proceeds as in~\eqref{eq:pattern2Upper}. We then get the upper bound 
		\begin{align*}
		&\sup_\Phi \ell(\tree^0,\Phi,x)\nonumber\\
		&\le L \log\left(\frac{1}{2}\right) + 2 \log \left(\frac{1}{L}\right) + 2 (L-1) \log  \left(1-\frac{1}{L}\right).
		\end{align*}
		On the other hand, for $\tree^*$ (or, in fact, any topology), setting $\lambda_b=\lambda_d=\lambda_m=0$ and $\lambda_a = \lambda_c$ so that $p_a=p_c=\frac{1}{L}$, we get
		\begin{eqnarray*}
			&&\P_{x \sim (\tree^*,\Phi)}\left(x^a_1{=}1,x^b_1{=}x^c_1{=}x^d_1{=}0\right)\\
			&&\qquad \times \P_{x \sim (\tree^*,\Phi)}\left(x^c_1{=}1,x^a_1{=}x^b_1{=}x^d_1{=}0\right)\\
			&& \qquad \times\P_{x \sim (\tree^*,\Phi)}\left(x^a_i{=}x^b_i{=}x^c_i{=}x^d_i{=}0\right)^{L-2}\\
			&& = \frac{1}{2}\left(\frac{1}{L}\right)\left(1 - \frac{1}{L}\right)
			\times \frac{1}{2}\left(\frac{1}{L}\right) \left(1 - \frac{1}{L}\right)\nonumber\\
			&&\qquad \times
			\left[\frac{1}{2}\left(1 - \frac{1}{L}\right)^2\right]^{L-2},
		\end{eqnarray*}
		which establishes the required lower bound
		on $\sup_\Phi \ell(\tree^*,\Phi,x)$.

		\item {\it One site with a 2/2-split $ac|bd$:}
		Without loss of generality, we assume that $(x^a_1,x^b_1,x^c_1,x^d_1)=(1,0,1,0)$
		and $(x^a_i,x^b_i,x^c_i,x^d_i)=(0,0,0,0)$ 
		for all $i = 2,\ldots,L$. Arguing as in the previous
		case,
		\begin{eqnarray*}
			&&\P_{x \sim (\tree^0,\Phi)}\left(x^a_1{=}x^c_1{=}1,x^b_1{=}x^d_1{=}0\right)\\
			&& \qquad \times\P_{x \sim (\tree^0,\Phi)}\left(x^a_i{=}x^b_i{=}x^c_i{=}x^d_i{=}0\right)^{L-1}\\
			&& \le \P_{x \sim (\tree^0,\Phi)}\left(x^a_1{=}1{\ne}x^b_1,x^c_1{\ne}x^d_1\right)\\
			&&  \qquad \times \P_{x \sim (\tree^0,\Phi)}\left(x^a_i{=}0{=}x^b_i,x^c_i{=}x^d_i\right)^{L-1}\\
			&& = [\P_{x \sim (\tree^0,\Phi)}\left(x^a_1{=}1{\ne}x^b_1\right)\,\P_{x \sim (\tree^0,\Phi)}\left(x^c_1{\ne}x^d_1\right)]\\
			&& \qquad \times [\P_{x \sim (\tree^0,\Phi)}\left(x^a_i{=}0{=}x^b_i\right)
			\, \P_{x \sim (\tree^0,\Phi)}\left(x^c_i{=}x^d_i\right)]^{L-1}\\
			&& = \frac{1}{2}\left(\frac{1-\phi_{ab}}{2}\right)
			\left(\frac{1-\phi_{cd}}{2}\right)\\
			&&  \qquad  \times \left[\frac{1}{2}\left(\frac{1+\phi_{ab}}{2}\right)\left(\frac{1+\phi_{cd}}{2}\right)\right]^{L-1}\\	
			&& = \left(\frac{1}{2}\right)^L
			\left(\frac{1-\phi_{ab}}{2}\right)
			\left(\frac{1+\phi_{ab}}{2}\right)^{L-1}\nonumber\\
			&& \qquad \left(\frac{1-\phi_{cd}}{2}\right)
			\left(\frac{1+\phi_{cd}}{2}\right)^{L-1},
		\end{eqnarray*}	
		where, again, $\phi_{ab} = \phi_a \phi_b$ and
		$\phi_{cd} = \phi_c \phi_d$.
		This bound matches the bound we obtained in the
		previous case. Hence, we once again get the upper bound 
		\begin{align*}
		&\sup_\Phi \ell(\tree^0,\Phi,x)\\
		&\le L \log\left(\frac{1}{2}\right) + 2 \log \left(\frac{1}{L}\right) + 2 (L-1) \log  \left(1-\frac{1}{L}\right).
		\end{align*}
		However, in this case, we claim that the maximum likelihood under $\tree^*$ is strictly greater. Indeed, letting $\lambda_a=\lambda_b=\lambda_c=\lambda_d=0$ and setting $\lambda_m$ such that $p_m=\frac{1}{L}$, 
		we get
		\begin{eqnarray*}
			&&\P_{x \sim (\tree^*,\Phi)}\left(x^a_1{=}x^c_1{=}1,x^b_1{=}x^d_1{=}0\right)\\
			&& \qquad \times\P_{x \sim (\tree^*,\Phi)}\left(x^a_i{=}x^b_i{=}x^c_i{=}x^d_i{=}0\right)^{L-1}\\		
			&&= \frac{1}{2}\left(\frac{1}{L}\right)
			\times
			\left[\frac{1}{2}\left(1 - \frac{1}{L}\right)\right]^{L-1},
		\end{eqnarray*}
		so
		\begin{align*}
		&\sup_\Phi \ell(\tree^*,\Phi,x)\nonumber\\
		&\ge L \log\left(\frac{1}{2}\right) + \log \left(\frac{1}{L}\right) + (L-1) \log  \left(1-\frac{1}{L}\right).
		\end{align*}
		Therefore
		\begin{align*}
		&\sup_\Phi \ell(\tree^*,\Phi,x)
		- \sup_\Phi \ell(\tree^0,\Phi,x)\nonumber\\
		&\ge - \log \left(\frac{1}{L}\right) - (L-1) \log  \left(1-\frac{1}{L}\right)\nonumber\\ 
		&=: K_{12}> 0,
		\end{align*}
		where the last equality is a definition.
	\end{enumerate}
	In all the above cases, a similar argument still applies if one replaces $\tree^0$ with $\tree^1$ (by exchanging the roles of $b$ and $d$ throughout).
	That concludes the proof of the claim in (b).
\end{proof}

\begin{proof}[Proof of Lemma~\ref{lem:key2}]
	(a) The expressions for $\qsat_{ac|bd}$,
	$\qsat_{ab|cd}$ and $\qsat_{ad|bc}$ come
	from taking $L=1$ in Lemma~\ref{lem:key1} (a) and raising to the power $L$. Specifically, it was shown
	in~\eqref{eq:q12} that observing a single site splitting $a,c$ from $b,d$ has probability of the form $(1/2)\rho^2 + O(\rho^3)$. Since a saturated locus contains $L$ sites with the same probability, we raise this expression to the power $L$ to obtain 
	\begin{align*}
	\qsat_{ac|bd} =
	\left(\frac{1}{2}\right)^{L} \rho^{2L} + \bigo(\rho^{2L+1}).
	\end{align*}
	Similarly, it was observed in the proof of Lemma~\ref{lem:key1} (a) that observing a single site splitting $a,b$ from $c,d$ (or $a,d$ from $b,c$) has probability $O(\rho^3)$. Raising to the power $L$ gives $\qsat_{ab|cd}, \qsat_{ad|bc}
	= \bigo(\rho^{3L})$.
	
	For the first part of the claim, we consider several cases.
	\begin{itemize}
		\item[-] Suppose that sequence dataset $x$ contains at least one uninformative character (i.e., a constant site or a singleton). Then, in computing bootstrap supports, there is probability at least $(1/L)^L$ of resampling a dataset containing {\em only} that particular uninformative site.
		We have shown in the proof of Lemma~\ref{lem:key1} (b) (see cases 1 and 2 with $L=1$) that all topologies have an equal ML score on such a site and therefore on such a resampled dataset (since the probability of observing a dataset of this type is the probability of observing a single site to the power $L$). Hence each topology is supported with probability $1/3$. Hence the bootstrap support for the ML-optimizer for $x$ is at most $1-(2/3)(1/L)^L \leq B$ and $x$ is rejected by WSB*.
		
		\item[-] Suppose that sequence dataset $x$ contains two different informative characters (i.e., two different splits). One of those splits is incompatible with the ML-optimizer (possibly random) for $x$. Then, in computing bootstrap supports, there is probability at least $(1/L)^L$ of resampling a dataset containing only that incompatible split.  From the argument in Lemma~\ref{lem:key1} (b) again (case 5 with $L=1$), the incompatible split is then the ML-optimizer of such a resampled dataset. Hence the bootstrap support for the ML-optimizer for $x$ is at most $1-(1/L)^L < 1-(2/3)(1/L)^L \leq B$ and $x$ is rejected by WSB*.
		
		\item[-] Suppose finally that sequence dataset $x$ contains only characters equivalent to a given split. Then all resampled datasets are saturated for that split as well.  From the argument in Lemma~\ref{lem:key1} (b) again (case 5 with $L=1$), that split is the unique ML-optimizer for $x$. Hence the bootstrap support for the ML-optimizer for $x$ is $1 > B$ and $x$ is passed along by WSB* to the summary method.
	\end{itemize}
	
	\medskip
	
	(b) This was proved in (a).
	
\end{proof}

\subsection*{Partitioned ML on CFN model}

\begin{proof}[Proof of Claim~\ref{claim:partML}]
	Using Lemma~\ref{lem:key1}, we are now ready to prove Claim~\ref{claim:partML}.
	
	We first show that, for a fixed topology, as the number of loci grows to infinity the maximum likelihood value converges almost surely to the expected value of the maximum likelihood value on a single locus.
	\begin{lem}[Convergence of the partitioned log-likelihood]\label{lem:conv-like}
		Let $\tree'$ be a fixed topology on the four taxa with branch lengths $\Lambda'$. Let also $\tree''$ be a fixed topology on the four taxa (possibly, but not necessarily, equal to $\tree'$). If the length-$L$ locus sequence datasets $\chi_{\cdot j}$, $j =1,2,\ldots$, are generated under the CFN model on $(\tree',\Lambda')$, then it holds that
		\begin{align}
		&\frac{1}{m}\sum_{j=1}^m \sup_{\Lambda_j}  \ell(\tree'',\Lambda_j,\chi_{\cdot j})\nonumber\\
		&\rightarrow \E_{\chi_{\cdot 1} \sim (\tree',\Lambda')} \left[\sup_\Lambda \ell(\tree'',\Lambda,\chi_{\cdot 1})\right]
		\in [-4 L \log 2, - L \log 2],\label{likelihoodConvergesAS}
		\end{align}
		almost surely as $m\rightarrow+\infty$. Above, the subscript $\chi_{\cdot 1} \sim (\tree',\Lambda')$ indicates that the expectation
		is taken over a single locus under the CFN model on $(\tree',\Lambda')$.
	\end{lem}
	\begin{proof}
		For a given topology and data set there is a unique maximum likelihood value, though the branch lengths at which it is attained may not themselves be unique.
		For any given locus $j$, there are a finite number of four-sequence data sets $\chi_{\dot j}$ of length $L$ that can occur under the CFN model. As the number of loci approaches infinity, the frequency of each data set approaches its expected value by the Strong Law of Large Numbers (SLLN) (see, e.g.,~\cite{Durrett:96}).
		To check that the conditions of
		the SLLN are satisfied, note that the log-likelihood is
		non-positive. In fact, by taking branch lengths to $+\infty$ under the CFN model, we have for any topology $\tree$ on $\{a,b,c,d\}$ and any locus data set $\chi_{\cdot 1}$
		\begin{equation}
		\label{logLLowerbound1}
		\sup_{\Lambda} \ell(\tree,\Lambda,\chi_{\cdot 1})
		\geq \log \left(\frac{1}{2}\right)^{4L}
		= -4 L \log 2.
		\end{equation}
		On the other hand, because by inclusion the probability
		of observing $\chi_{\cdot 1}$ is at most the probability
		of observing $\chi_{a 1}$, which is simply $\left(\frac{1}{2}\right)^L$ by independence of the sites, we also have
		\begin{equation}
		\label{logLLowerbound2}
		\sup_{\Lambda} \ell(\tree,\Lambda,\chi_{\cdot 1})
		\leq \log \left(\frac{1}{2}\right)^{L}
		= - L \log 2.
		\end{equation}
		So the expectation on the RHS of~\eqref{likelihoodConvergesAS} lies in the
		interval $[-4 L \log 2, - L \log 2]$.  
	\end{proof}

	Hence, in view of Lemma~\ref{lem:conv-like}, our goal is to show that there is $\rho > 0$ small enough such that the expected log-likelihood under $(\tree^0,\Lambda^0)$ is higher for $\tree^*$ than it is for $\tree^0$ or $\tree^1$. 
	That is, it suffices to establish the following
	claim.
	\begin{lem}[Expected locus-wise maximum likelihood on a fixed topology: key inequality]\label{lem:exp-ml}
		There exists $\rho > 0$ 
		such that
		\label{claim:expectations}
		\begin{align}
		&\E_{\chi_{\cdot 1} \sim (\tree^0,\Lambda^0)} \left[\sup_\Lambda \ell(\tree^0,\Lambda,\chi_{\cdot 1})\right]\nonumber\\
		&<
		\E_{\chi_{\cdot 1} \sim (\tree^0,\Lambda^0)} \left[\sup_\Lambda \ell(\tree^*,\Lambda,\chi_{\cdot 1})\right],\label{likelihoodExpectation}
		\end{align} 
		and
		\begin{align}
		&\E_{\chi_{\cdot 1} \sim (\tree^0,\Lambda^0)} \left[\sup_\Lambda \ell(\tree^1,\Lambda,\chi_{\cdot 1})\right]\nonumber\\
		&<
		\E_{\chi_{\cdot 1} \sim (\tree^0,\Lambda^0)} \left[\sup_\Lambda \ell(\tree^*,\Lambda,\chi_{\cdot 1})\right].\label{likelihoodExpectation-1}
		\end{align} 
		
	\end{lem}
	\begin{proof}
		Let $\patterns$ be the set of all possible
		single-locus datasets. 
		To prove Lemma~\ref{claim:expectations}, 
		we expand the expectations in~\eqref{likelihoodExpectation} over $\patterns$.
		In other words, we seek to show that
		\begin{align}
		&\sum_{x \in \patterns} 
		\P_{\chi_{\cdot 1} \sim (\tree^0,\Lambda^0)}[\chi_{\cdot 1} = x]\nonumber\\
		& \qquad \times \left\{\sup_\Lambda \ell(\tree^*,\Lambda,x) - \sup_\Lambda \ell(\tree^0,\Lambda,x)\right\} > 0. \label{likelihoodExpectation2}	
		\end{align}
		We then use Lemma~\ref{lem:key1} as follows. By (a), 	
		\begin{align}
		&\sum_{x \in \widetilde{\patterns}} 
		\P_{\chi_{\cdot 1} \sim (\tree^0,\Lambda^0)}[\chi_{\cdot 1} = x]\nonumber\\
		&\qquad \times \left|\,\sup_\Lambda \ell(\tree^*,\Lambda,x) - \sup_\Lambda \ell(\tree^0,\Lambda,x)\,\right| = \bigo(\rho^3). 	\label{likelihoodExpectation3}
		\end{align}
		Indeed, any locus site pattern in $\widetilde{\patterns}$ has probability $\mathcal{O}(\rho^3)$. Moreover, recall
		from~\eqref{logLLowerbound1} and~\eqref{logLLowerbound2} that the expression
		in absolute value is bounded by $3 L \log 2$. In addition, by (a) and (b), we then 
		arrive at
		\begin{align*}
		&\sum_{x \in \patterns} 
		\P_{\chi_{\cdot 1} \sim (\tree^0,\Lambda^0)}[\chi_{\cdot 1} = x]\nonumber\\
		& \qquad \times \left\{\sup_\Lambda \ell(\tree^*,\Lambda,x) - \sup_\Lambda \ell(\tree^0,\Lambda,x)\right\}\\
		&\geq K_{12} \left\{\left(\frac{1}{2}\right)^L \rho^2 + \bigo(\rho^3)\right\}
		+ \bigo(\rho^3) > 0,
		\end{align*}
		for $\rho > 0$ small enough.
		
		The same argument applies for~\eqref{likelihoodExpectation-1}. 	
	\end{proof}	
	
	Combining Lemmas~\ref{lem:conv-like}
	and~\ref{lem:exp-ml} gives Claim~\ref{claim:partML}.
\end{proof}

\subsection*{Reasonable summary methods on CFN model}

\begin{proof}[Proof of Claim~\ref{claim:reas}]
	Using Lemma~\ref{lem:key1}, we are now ready to prove Claim~\ref{claim:reas}.
	
	By definition of a reasonable summary method, on a four-taxon dataset, $\mathcal{A}$ outputs the most common quartet topology (breaking ties uniformly at random). We also assume that for genes with multiple optimal ML topologies, a highest scoring topology is picked uniformly at random. We denote by $\hat{\mathcal{R}}(\chi_{\cdot j})$ be the ML gene tree on the $j$-th locus sequence dataset. The law of large numbers immediately gives the following.
	\begin{lem}[Convergence of frequencies]\label{lem:conv-reas}
		Let $\tree'$ be a fixed topology on the four taxa with branch lengths $\Lambda'$. Let also $\tree''$ be a fixed topology on the four taxa (possibly, but not necessarily, equal to $\tree'$). If the length-$L$ locus sequence datasets $\chi_{\cdot j}$, $j =1,2,\ldots$, are generated under the CFN model on $(\tree',\Lambda')$, then it holds that
		\begin{equation*}
		\frac{1}{m}\sum_{j=1}^m 
		\mathbf{1}\left[
		\hat{\mathcal{R}}(\chi_{\cdot j}) = \tree''
		\right] 
		\longrightarrow 
		\P_{\chi_{\cdot 1} \sim (\tree',\Lambda')} \left[
		\hat{\mathcal{R}}(\chi_{\cdot 1}) = \tree''\right],
		\end{equation*}
		almost surely as $m\rightarrow+\infty$. Above, $\mathbf{1}[\mathcal{E}]$ is $1$ if event $\mathcal{E}$ occurs, and $0$ otherwise.
	\end{lem}
	
	Hence, in view of Lemma~\ref{lem:conv-reas}, our goal is to show that there is $\rho > 0$ small enough such that, under $(\tree^0,\Lambda^0)$, $\tree^*$ is more likely to be the ML gene tree topology than $\tree^0$ or $\tree^1$. 
	That is, it suffices to establish the following
	claim.
	\begin{lem}[Locus-wise maximum likelihood on a fixed topology: key inequality]\label{lem:exp-reas}
		There exists $\rho > 0$ 
		such that
		\begin{align}
		&\P_{\chi_{\cdot 1} \sim (\tree^0,\Lambda^0)} \left[\hat{\mathcal{R}}(\chi_{\cdot 1}) =\tree^0\right]\nonumber\\
		&<
		\P_{\chi_{\cdot 1} \sim (\tree^0,\Lambda^0)} \left[\hat{\mathcal{R}}(\chi_{\cdot 1}) =\tree^*\right],\label{likelihoodExpectation-reas}
		\end{align} 
		and
		\begin{align}
		&\P_{\chi_{\cdot 1} \sim (\tree^0,\Lambda^0)} \left[\hat{\mathcal{R}}(\chi_{\cdot 1}) =\tree^1\right]\nonumber\\
		&<
		\P_{\chi_{\cdot 1} \sim (\tree^0,\Lambda^0)} \left[\hat{\mathcal{R}}(\chi_{\cdot 1}) =\tree^*\right].\label{likelihoodExpectation-reas-1}
		\end{align} 
		
	\end{lem}
	\begin{proof}
		By Lemma~\ref{lem:key1} (b), for all $x \in \patterns_0 \cup \patterns_{11} \cup \patterns_{2=} \cup \patterns_{2\neq}$,
		all three topologies are ML-optimal, while
		for all $x \in \patterns_{12}$, $\tree^*$ alone
		is ML-optimal. Moreover, by Lemma~\ref{lem:key1} (a), all other patterns
		are negligible. Hence, we get
		\begin{align*}
		&\P_{\chi_{\cdot 1} \sim (\tree^0,\Lambda^0)} \left[\hat{\mathcal{R}}(\chi_{\cdot 1}) =\tree^*\right]\\
		&\geq \frac{1}{3} 2^L
		[Q_0 + 4 Q_{11} + 4 Q_{2=} + 3 Q_{2 \neq}]
		+ 2^L Q_{12}
		+ \bigo(\rho^3),
		\end{align*}
		while
		\begin{align*}
		&\P_{\chi_{\cdot 1} \sim (\tree^0,\Lambda^0)} \left[\hat{\mathcal{R}}(\chi_{\cdot 1}) =\tree^0\right]\\
		&\leq \frac{1}{3} 2^L
		[Q_0 + 4 Q_{11} + 4 Q_{2=} + 3 Q_{2 \neq}]
		+ \bigo(\rho^3),
		\end{align*}
		and similarly for $\tree^1$.
		The result then follows from the fact that
		\begin{equation*}
		Q_{12}=  \left(\frac{1}{2}\right)^L \rho^2 + \bigo(\rho^3).
		\end{equation*}
	\end{proof}
	Combining Lemmas~\ref{lem:conv-reas}
	and~\ref{lem:exp-reas} gives Claim~\ref{claim:reas}.
\end{proof}

\subsection*{WSB* pipeline on CFN model}

\begin{proof}[Proof of Claim~\ref{claim:wsb*}]
	Using Lemma~\ref{lem:key2}, we are now ready to prove Claim~\ref{claim:wsb*}.
	
	We begin with two basic results.
	\begin{lem}
		\label{lem:bins}
		In a WSB* pipeline with bootstrap
		support threshold $B \geq 1 - \frac{2}{3}\left(\frac{1}{L}\right)^L$, 
		there will be at most three bins (one for each of the three
		possible binary topologies on four leaves), and
		the bin associated with topology $ab|cd$ will have
		all the saturated genes that support $ab|cd$ (and similarly for $ac|bd$ and $ad|bc$).
	\end{lem}
	\begin{proof}
		By Lemma~\ref{lem:key2} (b), the 
		genes saturated for a given split have bootstrap support of 100\%. Hence, no
		two such genes can be placed in the same bin if
		they support different tree topologies. Therefore,
		for any bin, 
		the genes placed in the bin will support the
		same topology.
		By Lemma~\ref{lem:key2} (a), all other genes are discarded.
		
		Since there are only three tree topologies,  the incompatibility graph 
		is the union of a complete 3-partite graph (defined by the split-saturated
		genes). 
		Hence, the incompatibility graph
		can be $3$-colored. 
		Since statistical binning seeks the minimum vertex coloring
		for the incompatibility graph, it will partition the
		genes into three bins, with one bin for each binary tree topology. 
		Hence, the split-saturated genes are partitioned into three sets
		based on the tree topology they support.
	\end{proof}
	\begin{lem}
		\label{lemma2}
		(Lemma 2 from \cite{WSB}:)
		Let $S$ be a set of taxa, and let 
		$S_i$ be a set of DNA sequences for $S$, with 
		$i = 1, 2, \ldots p$.  
		Suppose that tree topology $t$ is an optimal 
		solution for GTR maximum likelihood for each $S_i$
		(allowing various GTR parameters 
		for different $i = 1, 2, \ldots p$).  
		Then $t$ will be an optimal solution
		to a fully partitioned GTR maximum likelihood analysis 
		on a concatenation of $S_1, S_2, \ldots, S_p$.
	\end{lem}
	
	\begin{cor}
		\label{corollary-to-lemma2}
		The set of newly computed gene trees computed during a WSB* pipeline 
		has the same distribution as the original set of ML gene trees
		obtained from the split-saturated genes.
	\end{cor}
	\begin{proof}
		By Lemma \ref{lem:bins},
		the split-saturated
		genes are partitioned into three
		bins for the different tree topologies.
		By Lemma \ref{lemma2}, 
		fully
		partitioned maximum likelihood
		on each supergene alignment produces the
		tree topology associated
		with the bin.
		In a WSB* pipeline, the supergene tree for each bin is copied
		by as many genes as in the bin.
		Hence, the distribution defined by the newly
		computed gene trees is identical to the
		distribution defined by  
		original ML gene trees.
	\end{proof}
	
	The rest of the argument follows as in the proof of Claim~\ref{claim:reas}.
	By Lemma~\ref{lem:key2} (a), under our four-taxon model species tree with topology $ab|cd$, 
	the most probable estimated quartet tree on split-saturated genes is $ac|bd$.
	After removing all the loci that are not split-saturated, we are left only 
	with genes that split $2/2$. As the number of loci increases, with probability going to $1$
	the most frequent estimated quartet tree will be $ac|bd$.
	Therefore by Corollary~\ref{corollary-to-lemma2}, in a WSB* pipeline with bootstrap
	support threshold $B \geq 1 - \frac{2}{3}\left(\frac{1}{L}\right)^L$, 
	the most frequent supergene tree computed
	by weighted statistical binning is identical to the most frequent
	estimated quartet tree in the input,
	and will converge to $ac|bd$ as the number of loci increases by the law of large numbers.
	Hence, WSB* pipelines 
	followed by reasonable summary methods will be 
	positively misleading under this model.
\end{proof}

\subsection*{Extension to MSC+CFN model}

In this section, we extend the main claims to the MSC+CFN model. The key idea is to choose a species tree that is highly likely to produce, on any given locus, sequence data whose distribution is close to that of a fixed gene tree in the Felsenstein zone.

When a character of length $L$, $\chi_{\cdot j}$, is generated under the CFN model on $(\tree,\Lambda)$,
we write $\chi_{\cdot j} \sim \mathcal{D}_g^L[\tree,\Lambda]$.
Formally, $\mathcal{D}_g^L[\tree,\Lambda]$ is a probability distribution over sequence datasets in$\{0,1\}^{n\times L}$, that is, containing $n$ sequences of length $L$ taking values in $\{0,1\}$, where
$n$ is the number of leaves in $\tree$. The subscript $g$ is meant to refer to the fact that this is a distribution obtained from a single gene tree.

We also consider sequence datasets generated by the MSC+CFN model.
Consider a species tree 
$(\stree,\Gamma,\theta)$
with 
$n$ leaves. Each gene $j = 1,\ldots, m$ 
has a genealogical history 
represented by its gene tree $\tree_j$
distributed
according to the following process: 
looking backwards in time, 
on each branch $e$ of the species tree, 
the coalescence of any two lineages 
is exponentially distributed with rate $2/\theta_e$, 
independently from all other pairs; 
whenever two branches merge in the species tree, 
we also merge the lineages of the corresponding populations, that is, the coalescence proceeds on the \emph{union} of the lineages. 
More specifically, the probability density of a realization of this model for $m$ independent genes is
\begin{align*}
&\prod_{j=1}^m \prod_{e\in E} 
\exp\left(-\binom{O_j^{e}}{2} 
\left[\gamma_j^{e, O_j^{e}+1} - \gamma_j^{e, O_j^{e}}\right]\frac{2}{\theta_e}\right)\\
&\qquad\qquad \times\prod_{\ell=1}^{I_j^{e}-O_j^{e}}
\exp\left(-\binom{\ell}{2} 
\left[\gamma_j^{e, \ell} - \gamma_j^{e, \ell-1}\right]\frac{2}{\theta_e}\right),
\end{align*}
where, for gene $j$ and branch $e$, 
$I_j^{e}$ is the number of lineages entering $e$,
$O_j^{e}$ is the number of lineages exiting $e$, and
$\gamma_j^{e,\ell}$ is the $\ell^{th}$ coalescence time in $e$;
for convenience, we let $\gamma_j^{e,0}$ and
$\gamma_j^{e,I_j^{e}-O_j^{e}+1}$ be respectively
the divergence times of $e$ and of its parent population (which depend on $\Gamma$). 

When a character of length $L$, $\chi_{\cdot j}$, is generated under the MSC+CFN model on $(\stree,\Gamma,\theta)$,
we write $\chi_{\cdot j} \sim \mathcal{D}_s^L[\stree,\Gamma,\theta]$.
Formally, $\mathcal{D}_s^L[\stree,\Gamma,\theta]$ is a probability distribution over sequence datasets in $\{0,1\}^{n\times L}$, where
$n$ is the number of leaves in $\stree$. The subscript $s$ is meant to refer to the fact that this is a distribution obtained from the MSC on a species tree.

As in the main text, fix $\tree^0$ to be the four-taxon topology $ab|cd$ on $\{a,b,c,d\}$ and let $\Lambda^0$ denote a vector of branch lengths on $\tree^0$. Denote the endpoint of the middle edge on the $ab$ side as $e$, and on the $cd$ side as $f$. For this tree, denote the length of branch $ae$ as $\lambda^0_a$, $be$ as $\lambda^0_b$, $cf$ as $\lambda^0_c$, $df$ as $\lambda^0_d$ and $ef$ as $\lambda^0_m$.
For a branch length $\lambda$, recall that we also
use the parametrization $\phi 
= -\frac{1}{2} \log \lambda$ in terms of which
the probability of a change along this branch is
\begin{equation*}
p=\frac{1}{2}\left(1-e^{-2\lambda}\right)=\frac{1}{2}(1-\phi),
\end{equation*}
and the probability of no change is
$q=\frac{1}{2}(1+\phi).$
We choose $\Lambda^0$ to construct a Felsenstein zone tree where, for a parameter $\rho > 0$, $p^0_a=p^0_c=\rho$ and $p^0_b=p^0_d=p^0_m=\rho^3$. Note that for any $\rho>0$, we can set $\lambda^0_a=\lambda^0_c=-\frac{1}{2}\log (1-2\rho)$ and $\lambda^0_b=\lambda^0_d=\lambda^0_m=-\frac{1}{2}\log (1-2\rho^3)$ to satisfy this relationship. We also denote the alternate topologies
by $\tree^*=ac|bd$ and $\tree^1=ad|bc$.
\begin{claim}[Species tree in the Felsenstein zone]\label{claim:species-fels}
	For all $\epsilon > 0$, there is a species tree $(\stree^0,\Gamma^0,\theta^0)$
	with leaves $\{a,b,c,d\}$ and a probability distribution $\mathcal{R}$ over $\{0,1\}^{4\times L}$ such that
	\begin{align*}
	\mathcal{D}_s^L[\stree^0,\Gamma^0,\theta^0]
	= 
	(1-\epsilon)\, \mathcal{D}_g^L[\tree^0,\Lambda^0]
	+ \epsilon\, \mathcal{R}.
	\end{align*}
\end{claim}
\begin{proof}
	We let $\stree^0$ be the balanced species tree with split $ab|cd$ and root $r$.
	Denote the endpoint of the edge incident to the root on the $ab$ side as $e$, and on the $cd$ side as $f$. 
	For this tree, denote the length of branch $ae$ as $\gamma^0_a$, $be$ as $\gamma^0_b$, $cf$ as $\gamma^0_c$, $df$ as $\gamma^0_d$, $er$ as $\gamma^0_e$ and $fr$ as $\gamma^0_f$.
	And similarly for $\theta^0$. The branch $r\infty$ above the root $r$ has infinite length and parameter $\theta^0_r$. We take $\theta^0_{a} = \theta^0_b = \theta^0_c
	= \theta^0_d = 1$, $\gamma^0_a = \lambda^0_a$, $\gamma^0_b = \lambda^0_b$, $\gamma^0_c = \lambda^0_c$, $\gamma^0_d = \lambda^0_d$. Finally we let
	$\gamma^0_e = \gamma^0_f = \alpha + \lambda^0_m/2$ and $\theta^0_e = \theta^0_f = \theta^0_r = \beta$. Take $\alpha$ and $\beta$ small enough that:
	\begin{itemize}
		\item coalescences in $er$, $fr$ and $r\infty$ occur within $\alpha$ of $e$ $f$, and $r$ respectively;
		
		\item no mutation occurs within $\alpha$ above $e$, $f$ and $r$ respectively;
	\end{itemize}
	with probability at least $1-\epsilon$. Conditioned on the event above, the distribution of sequence dataset is precisely
	$\mathcal{D}_g^L[\tree^0,\Lambda^0]$.
	The result follows.
\end{proof}

We are now ready to prove the main theorems.
\begin{proof}[Proof of Theorem~\ref{thm:incPart}]
	We take $(\stree^0,\Gamma^0,\theta^0)$
	as in Claim~\ref{claim:species-fels}
	for $\epsilon > 0$ to be determined below.
	We think of the first $m$ loci as divided into two subsets: $\mathcal{M}^m_0$ coming from distribution  $\mathcal{D}_g^L[\tree^0,\Lambda^0]$
	and $\mathcal{M}^m_{\mathcal{R}}$ coming from $\mathcal{R}$. By the law of large numbers, we have
	\begin{align*}
	\frac{|\mathcal{M}^m_0|}{m} \to 1-\epsilon
	\qquad\text{and}\qquad  
	\frac{|\mathcal{M}^m_{\mathcal{R}}|}{m} \to \epsilon.
	\end{align*}
	We then apply the argument in the proof of Claim~\ref{claim:partML} to the samples in $\mathcal{M}^m_0$ and take 
	$\epsilon$ small enough that the contribution of $\mathcal{M}^m_{\mathcal{R}}$ 
	to the partitioned log-likelihood 
	is in the limit $m \to +\infty$ smaller
	than the expected gap between $\tree^*$ and $\tree^0$.
\end{proof}

The proofs of Theorems~\ref{thm:incReas} and~\ref{thm:incWSB*} follow from similar arguments.

\end{document}